\documentclass[11pt]{article}

\usepackage[utf8]{inputenc}
\usepackage[margin=1in]{geometry}
\usepackage{amsmath,amssymb,amsfonts,amsthm,mathtools}
\usepackage{bm}
\usepackage{url}
\usepackage{booktabs}
\usepackage{graphicx}
\usepackage{algorithm}
\usepackage{algpseudocode}
\usepackage{threeparttable}
\usepackage[figuresright]{rotating}
\usepackage[authoryear,round]{natbib}
\usepackage{caption}
\usepackage{subcaption}
\usepackage{xcolor}
\usepackage{setspace}
\usepackage{array}
\usepackage{multirow}
\usepackage{longtable}
\usepackage[colorlinks=true,linkcolor=blue,citecolor=blue,urlcolor=blue]{hyperref}
\usepackage[nameinlink,noabbrev]{cleveref}

\newcommand{\tr}{\operatorname{tr}}
\newcommand{\vecop}{\operatorname{vec}}

\newcommand{\diag}{\operatorname{diag}}

\newcommand{\Cov}{\operatorname{Cov}}

\newcommand{\I}{\mathbf{I}}
\newcommand{\A}{\mathbf{A}}
\newcommand{\B}{\mathbf{B}}
\newcommand{\Q}{\mathbf{Q}}
\newcommand{\rank}{\operatorname{rank}}
\newcommand{\one}{\mathbf{1}}

\newcommand{\bb}{\mathbf{b}}
\newcommand{\D}{\mathbf{D}}
\newcommand{\V}{\mathbf{V}}
\newcommand{\grad}{\operatorname{grad}}

\newtheorem{proposition}{Proposition}
\newtheorem{lemma}{Lemma}
\newtheorem{theorem}{Theorem}
\newtheorem{corollary}{Corollary}
\newtheorem{assumption}{Assumption}

\Crefname{table}{Web Table}{Web Tables}
\Crefname{figure}{Web Figure}{Web Figures}
\Crefname{algorithm}{Web Algorithm}{Web Algorithms}
\Crefname{section}{Web Appendix}{Web Appendix}
\Crefname{subsection}{Web Appendix}{Web Appendix}

\title{Hierarchical Probabilistic Principal Component Analysis of Longitudinal Data}

\author{
Xinyu Zhang \and Ameer Qaqish \and D.Y. Lin$^{*}$ \and Didong Li$^{*}$ \\
Department of Biostatistics, University of North Carolina at Chapel Hill\\
Chapel Hill, North Carolina, U.S.A.\\
*Corresponding emails: lin@bios.unc.edu, didongli@unc.edu
}

\date{\today}

\begin{document}

\maketitle

\begin{abstract}
In many longitudinal studies, a large number of variables are measured repeatedly over time, with substantial missing data. Existing methods, such as probabilistic principal component analysis (PPCA), are ill-equipped to handle such incomplete, high-dimensional longitudinal data, as they fail to account for the nested sources of variation and temporal dependency inherent in repeated measures. We introduce hierarchical probabilistic principal component analysis (HPPCA), a two-level probabilistic factor model that explicitly separates between-subject variance from time-varying within-subject dynamics. The within-subject latent factors are modeled by a Gaussian process. We develop an EM algorithm to handle missing data and flexible covariance kernels, accelerated by computationally efficient initializers. Simulation studies demonstrated that HPPCA robustly recovers model parameters subspaces and substantially outperforms both standard PPCA and multivariate functional PCA in imputation accuracy, even under heavy missingness and model misspecification. An application to the long COVID symptoms in the Researching COVID to Enhance Recovery adult cohort revealed that HPPCA effectively captured the data's hierarchical structure and its learned features significantly improved the prediction of clinical outcomes and the recovery of masked clinical records compared to exisiting methods.
\end{abstract}



\section{Introduction}

The analysis of high-dimensional longitudinal data, where a large number of variables are measured repeatedly over time on many subjects, is a central task in clinical and epidemiological research \citep{fitzmaurice2008longitudinal}. A key challenge lies in the complex correlation structure inherent in such data, which arises from the hierarchical nesting of observations within subjects and the temporal dependency of repeated measurements. Furthermore, these datasets are often characterized by substantial missingness and irregular, subject-specific visit schedules. A principled low–dimensional representation must therefore disentangle the between– and within–subject sources of variation, accommodate temporal correlation, and remain robust in the presence of incomplete observations.

Principal component analysis (PCA) is a cornerstone for dimension reduction, projecting correlated variables onto orthogonal directions that minimize reconstruction error \citep{Jolliffe2002Principal}. Despite its wide use, classical PCA lacks a probabilistic foundation and is not equipped to handle missing data. The development of probabilistic principal component analysis (PPCA) \citep{tipping1999probabilistic} provided a crucial generative model, casting principal axes as a maximum marginalized likelihood solution to a latent factor model and naturally accommodating missing data through the expectation-maximization (EM) algorithm \citep{dempster1977maximum, NIPS1997_d9731321}. 

However, the latent factors in PPCA are assumed to be independent and follow the standard normal distribution, limiting the model’s ability to represent data with intrinsic temporal or spatial correlation. The generalized PPCA (GPPCA) \citep{gu2020generalized} addresses this issue by modeling the latent factors as Gaussian processes (GPs), which provides a tailored solution for time-series or spatial data \citep{williams2006gaussian}. Yet, neither PPCA nor GPPCA targets the nested nature of longitudinal data in which total covariance arises from distinct sources at different levels. 

Related ideas have also been developed in the functional data analysis literature. Multivariate functional principal component analysis (mFPCA) has been used to characterize joint temporal variation across multiple outcomes \citep{chen2017quantifying,Happ2018multivariate}, while multilevel functional principal component analysis decomposes variation from repeated measurements into between-subject and within-subject components \citep{di2009multilevel,di2014multilevel}. A related probabilistic approach is dynamic PPCA \citep{nyamundanda2014dynamic}, which introduces temporal dependence through time-varying latent and residual variance processes, but does not target a hierarchical decomposition of between-subject and within-subject covariance under irregular observation schedules. Functional PCA approaches typically rely on the estimation of covariance or related inner-product structures with a series of basis functions, which can be difficult in settings with sparse and irregular follow-up and substantial item-level missingness \citep{yao2005functional}. Moreover, jointly accommodating multivariate and multilevel dependence remains methodologically challenging.

In this paper, we propose hierarchical probabilistic principal component analysis (HPPCA), extending PPCA to longitudinal or nested data. To make inference under irregular observation schedules and arbitrary item-level missingness, we develop an EM algorithm that  accommodates general temporal kernels. Furthermore, to accelerate convergence, we derive computationally efficient maximum likelihood estimators in a complete balanced-panel special case, which serve as initializers for the EM procedure. HPPCA jointly captures interpretable between‑subject traits and within‑subject dynamics while accommodating missingness and temporal structures within subjects. The latent space can be used for downstream analysis, and the model also provides a way to impute missing data using the posterior distribution.

The remainder of this paper is organized as follows. Section~\ref{sec:model} introduces the formulation of the HPPCA model, the EM algorithm and the initializers. Section~\ref{sec:simulation} compares the performance of HPPCA and three existing dimension reduction methods through extensive simulation studies under different settings. Section~\ref{sec:recover} describes the application of our model to the Researching COVID to Enhance Recovery (RECOVER) adult cohort. Section~\ref{sec:discussion} provides concluding remarks and discusses future directions.

\section{Methods}\label{sec:model}

\subsection{Hierarchical Probabilistic Principal Component Analysis (HPPCA)}
In longitudinal studies, data collection is often subject to irregular follow-up schedules and varying numbers of visits across subjects. We consider a longitudinal cohort study with $n$ subjects, where the $i$th subject is observed at a subject-specific set of survey occasions indexed by $\mathcal{O}_i$. Let $J_i = |\mathcal{O}_i|$ denote the total number of visits for the $i$th subject, and let $t_{ij}$ denote the specific observation time of survey $j \in \mathcal{O}_i$. 

For the $i$th subject at their $j$th survey, we observe a $p$-dimensional vector of centered features, denoted by $\bm{Y}_{ij} \in \mathbb{R}^p$. To explicitly disentangle the time-invariant between-subject variation from the time-varying within-subject dynamics, we propose the following HPPCA model:
\begin{equation}\label{eq:model}
  \bm{Y}_{ij} \;=\; \bm{W}_1 \bm{Z}_{1,i} + \bm{W}_2 \bm{Z}_{2,ij} + \sigma\,\bm{\varepsilon}_{ij},
\end{equation}
where $\bm{W}_1 \in \mathbb{R}^{p \times d_1}$ and $\bm{W}_2 \in \mathbb{R}^{p \times d_2}$ are the factor loading matrices at the subject level and subject-survey level, respectively ($d_1 \leq p$, $d_2 \leq p$). The error term is modeled as isotropic Gaussian noise, $\bm{\varepsilon}_{ij} \sim \mathcal{N}(\bm{0}, \bm{I}_p)$. The unobserved latent factors capture the nested hierarchy of the data. The $d_1$-dimensional subject-level latent traits $\bm{Z}_{1,i}$ represent static  characteristics and are assumed to follow a standard multivariate normal distribution $\mathcal{N}(\bm{0}, \bm{I}_{d_1})$. To model the temporal correlation of the $d_2$-dimensional subject-survey factors $\bm{Z}_{2,ij}$, we assume that for each latent dimension $r=1,\dots,d_2$, the subject-specific trajectory vector $\bm{z}_{2,i}^{(r)} = (Z_{2,ij}^{(r)})_{j \in \mathcal{O}_i}^\top \in \mathbb{R}^{J_i}$ follows a zero-mean Gaussian process (GP):
\begin{equation}\label{eq:gp_prior}
\bm{z}_{2,i}^{(r)} \sim \mathcal{N}\big(\bm{0}, \bm{\Sigma}_i^{(r)}\big).
\end{equation}
 The $(j, j')$th element of the temporal covariance matrix $\bm{\Sigma}_i^{(r)}$ can be parameterized with a kernel function
$(\bm{\Sigma}_i^{(r)})_{j,j'} =K(t_{i j}, t_{i j'}; \ell^{(r)})$, where $\ell^{(r)}$ is the kernel parameter for latent dimension $r$. The components $\bm{Z}_{1,i}$, $\bm{z}_{2,i}^{(r)}$, and $\bm{\varepsilon}_{ij}$ are mutually independent across all subjects and dimensions.

Define the positive semi-definite matrices $\bm{A}=\bm{W}_1\bm{W}_1^\top$ and $\bm{B}=\bm{W}_2\bm{W}_2^\top$, where $\text{rank}(\A) \leq d_1$ and $\text{rank}(\B) \leq d_2$. By stacking the observations for the $i$th subject into a single vector $\bm{y}_i = \operatorname{vec}([\bm{Y}_{ij}]_{j \in \mathcal{O}_i}) \in \mathbb{R}^{p J_i}$, the marginal distribution of $\bm{y}_i$ is multivariate normal with mean zero and the following covariance matrix:
\begin{equation}\label{eq:stacked-cov}
\operatorname{Cov}(\bm{y}_i) \;=\; \one_{J_i}\one_{J_i}^\top \otimes \bm{A} \;+\; \sum_{r = 1}^{d_2} \bm{\Sigma}_i^{(r)} \otimes (\bm{w}_{2r} \bm{w}_{2r}^\top) \;+\; \sigma^2 \bm{I}_{p J_i},
\end{equation}
where $\operatorname{vec}(\cdot)$ denotes the vectorization operator that stacks the columns of a matrix into a single column vector, $\otimes$ denotes the Kronecker product, $\bm{1}_{J_i}$ represents a $J_i$-dimensional vector of all ones, $\bm{I}_{p J_i}$ is the identity matrix of dimension $p J_i$, and $\bm{w}_{2r}$ is the $r$th column of the loading matrix $\bm{W}_2$.

\paragraph{Temporal Covariance Kernels}
To encode the decay of within-subject correlation over time, we consider two stationary kernels for the latent processes:
\begin{align*}
&\text{Radial Basis Function (RBF): } 
K_{\mathrm{RBF}}(t,t';\ell)\;=\;\exp\!\left(-\frac{(t-t')^2}{2\ell^2}\right), \qquad \ell>0,\\[2mm]
&\text{Mat\'ern } \nu=5/2: 
K_{5/2}(t,t';\ell)\;=\;\left(1+\frac{\sqrt{5}\,r}{\ell}+\frac{5r^2}{3\ell^2}\right)\exp\!\left(-\frac{\sqrt{5}\,r}{\ell}\right), 
\quad r=|t-t'|,\ \ell>0.
\end{align*}
The RBF kernel produces infinitely smooth latent trajectories, whereas the Mat\'ern-$5/2$ kernel accommodates twice differentiable dynamic fluctuations. The marginal variance of the GPs is fixed at $1$ to ensure parameter identifiability, absorbing the amplitude scale of the within-subject variation naturally into $\bm{W}_2$.

\paragraph{Identifiability}
As in standard latent factor models, the loading matrices are not individually identifiable. Accordingly, inference is expressed through covariance-level quantities. In particular, $
\bm{A}=\bm{W}_1\bm{W}_1^\top$ and $\bm{B}=\bm{W}_2\bm{W}_2^\top$ are identifiable with a shared kernel parameter $\ell$. 
When dimension-specific kernel parameters $\ell^{(r)}$ are allowed, the latent dimensions are identifiable only up to relabeling, so the parameters $\{(\ell^{(r)},\bm{w}_{2r}\bm{w}_{2r}^\top)\}_{r=1}^{d_2}$ are identified only up to permutation.

\subsection{EM Algorithm}
In longitudinal studies, item-level missingness is common. We assume an ignorable missingness mechanism for likelihood-based inference; specifically, data are missing at random (MAR) conditional on the observed entries \citep{rubin1976inference,little2019statistical}. For model fitting, we propose using the EM algorithm, which treats latent variables and missing entries in $\bm{Y}$ as missing data, produces numerically stable updates for the loadings and noise variance, and supports general positive‑definite temporal covariance matrices $\bm{\Sigma}(\ell)$. Write $\bm{\theta} = (\bm{W}_1, \bm{W}_2, \sigma^2, \{\ell^{(r)}\}_{r=1}^{d_2})$.

To deal with missingness, for each $(i,j)$, we split vectors and matrices by observed versus missing rows: $\bm{Y}_{ij}=\big(\bm{Y}_{ij}^{(o)\top},\ \bm{Y}_{ij}^{(m)\top}\big)^\top$. The loading matrices $\bm{W}_1$ and $\bm{W}_2$ are also permuted by rows into two parts respectively:
\[
\begin{pmatrix}
    \bm{Y}_{ij}^{(o)} \\
    \bm{Y}_{ij}^{(m)}
\end{pmatrix} = \begin{pmatrix}
    \bm{W}_1^{(o, ij)}\\
    \bm{W}_1^{(m, ij)}
\end{pmatrix} \bm{Z}_{1, i} +  \begin{pmatrix}
    \bm{W}_2^{(o, ij)}\\
    \bm{W}_2^{(m, ij)}
\end{pmatrix} \bm{Z}_{2, ij} + \sigma \bm{\varepsilon}_{ij}.
\]
 We define the complete latent state and complete observed design for the $i$th subject by stacking over surveys $j\in \mathcal{O}_i=\{j_1, \cdots, j_{J_i}\}$,
\[
\bm{Z}_i \;=\; \begin{bmatrix}\bm{Z}_{1,i}\\\bm{Z}_{2,ij_1} \\\vdots\\ \bm{Z}_{2,ij_{J_i}}\end{bmatrix}, \qquad
\bm{W}_i^{(o)} \;=\; \begin{bmatrix}
\bm{W}_1^{(o,ij_1)} & \bm{W}_2^{(o,ij_1)}\bm{S}_{i1}\\
\vdots & \vdots\\
\bm{W}_1^{(o,ij_{J_i})} & \bm{W}_2^{(o,ij_{J_i})}\bm{S}_{iJ_i}
\end{bmatrix},\qquad
\bm{Y}_i^{(o)} \;=\; \begin{bmatrix}\bm{Y}_{ij_1}^{(o)}\\ \vdots\\ \bm{Y}_{ij_{J_i}}^{(o)}\end{bmatrix},
\]
where $\bm{S}_{ih} = \bm{e}_h^\top \otimes \bm{I}_{d_2} \in \mathbb{R}^{d_2 \times d_2 J_i}$ is a matrix that selects the hth survey-specific latent vector from the stacked vector $\bm{Z}_{2,i} = (\bm{Z}_{2,ij_1}^\top, \cdots, \bm{Z}_{2,ij_{J_i}}^\top)^\top$, and $\bm{e}_h \in \mathbb{R}^{J_i}$ is the canonical basis vector for the $h$th observed survey, $h = 1, \cdots, J_i$ is the stacked index.

Up to additive constants independent of the parameters, the complete-data log-likelihood for the $i$th subject is
\begin{align}\label{eq:complete_lik}
    l_{i}(\bm{\theta}) =& -\frac{pJ_i}{2} \log(2\pi \sigma^2) - \frac{1}{2\sigma^2}\sum_{j \in \mathcal{O}_i}\|\bm{Y}_{ij} - \bm{W}_1 \bm{Z}_{1, i} - \bm{W}_2 \bm{Z}_{2, ij}\|_2^2 \nonumber \\
    & - \frac{1}{2} \bm{Z}_{1, i}^\top\bm{Z}_{1, i} - \frac{1}{2}\sum_{r = 1}^{d_2} \Big\{\log\det\bm{\Sigma}_i^{(r)} + \bm{z}_{2, i}^{(r)\top} (\bm{\Sigma}^{(r)}_i)^{-1} \bm{z}_{2, i}^{(r)} \Big\}.
\end{align}
The total complete data log-likelihood is $l(\bm{\theta}) = \sum_{i=1}^n l_{i}(\bm{\theta})$.

In the E-step, we compute the expected value of the complete-data log-likelihood given the observed data $\bm{Y}^{(o)} = \{\bm{Y}_i^{(o)}\}_{i=1}^n$ and current parameter estimates $\bm{\theta}^{(t)}$. We define the $Q$-function:
\begin{equation}
    Q(\bm{\theta} \mid \bm{\theta}^{(t)}) = \sum_{i=1}^n \mathbb{E}\big\{\, l_{i}(\bm{\theta}) \mid \bm{Y}_i^{(o)}, \bm{\theta}^{(t)}\,\big\}.
\end{equation}
To evaluate the $Q$-function, we require the posterior distribution of the latent variables, $\bm{Z}_i \mid \bm{Y}_{i}^{(o)}, \bm{\theta}^{(t)}$. Due to the conjugate Gaussian formulation, this posterior is multivariate normal, $\mathcal{N}(\bm{\mu}_{Z_i}, \bm{\Sigma}_{Z_i})$. 

Let $\bm{\Lambda}_i = \mathrm{blockdiag}\big(\bm{I}_{d_1}, \bm{\Omega}_{2,i}\big)$ be the block-diagonal precision matrix of the structural priors, where $\bm{\Omega}_{2,i} = \sum_{r=1}^{d_2} (\bm{\Sigma}_i^{(r)})^{-1} \otimes (\bm{e}_r \bm{e}_r^\top)$, with $\bm{e}_r \in \mathbb{R}^{d_2}$ being the canonical basis vector for latent dimension $r$. The exact posterior covariance matrix and mean vector are given by:
\begin{align}
    \bm{\Sigma}_{Z_i} &= \left\{ \bm{\Lambda}_i + \frac{1}{\sigma^2} \big(\bm{W}_i^{(o)}\big)^\top \bm{W}_i^{(o)} \right\}^{-1}, \label{eq:post_var} \\
    \bm{\mu}_{Z_i} &= \frac{1}{\sigma^2} \bm{\Sigma}_{Z_i} \big(\bm{W}_i^{(o)}\big)^\top \bm{Y}_i^{(o)}. \label{eq:post_mean}
\end{align}
We can then easily extract the expected sufficient statistics $\langle \bm{Z}_i \rangle = \bm{\mu}_{Z_i}$ and $\langle \bm{Z}_i \bm{Z}_i^\top \rangle = \bm{\mu}_{Z_i}\bm{\mu}_{Z_i}^\top + \bm{\Sigma}_{Z_i}$, where we use $\langle \rangle$ to denote the conditional expectation. In addition, we calculate the conditional expectations of the missing data $\langle \bm{Y}_{ij}^{(m)} \rangle$, which provide optimal model-based imputations for unobserved responses, alongside the cross-moment expectations $\langle \bm{Y}_{ij}^{(m)}\bm{Y}_{ij}^{(m)\top} \rangle$, $\langle \bm{Y}_{ij}^{(m)}\bm{Z}_{1,i}^\top \rangle$, and $\langle \bm{Y}_{ij}^{(m)}\bm{Z}_{2,ij}^{\top} \rangle$ (detailed in Web Appendix A).

In the M-step, we maximize $Q(\bm{\theta} \mid \bm{\theta}^{(t)})$ with respect to $\bm{\theta}$. Setting the derivatives of the $Q$-function with respect to $\bm{W}_1$ and $\bm{W}_2$ to zero yields closed-form generalized least-squares updates. 

Let $\widetilde{\bm{G}}_{ij}^{(1)} = \mathbb{E}[\bm{Y}_{ij} \bm{Z}_{1,i}^\top \mid \bm{Y}_i^{(o)}]$ and $\widetilde{\bm{G}}_{ij}^{(2)} = \mathbb{E}[\bm{Y}_{ij} \bm{Z}_{2,ij}^\top \mid \bm{Y}_i^{(o)}]$ represent the pseudo-data cross-moment matrices, where the missing components are evaluated through their exact conditional cross-moments, correctly capturing the posterior covariance between the missing features and the latent factors \citep{nyamundanda2010probabilistic}. The loading matrices are jointly updated by solving the augmented normal equations:
\begin{align}\label{eq:W_update}
\begin{pmatrix}
    \widetilde{\bm{W}}_1 & \widetilde{\bm{W}}_2
\end{pmatrix}  = 
\left( \sum_{i = 1}^n \sum_{j \in \mathcal{O}_i} 
\begin{bmatrix}
     \widetilde{\bm{G}}_{ij}^{(1)} & 
     \widetilde{\bm{G}}_{ij}^{(2)}
\end{bmatrix} \right)
 \left( \sum_{i=1}^n \sum_{j \in \mathcal{O}_i} 
 \left\langle \begin{bmatrix} \bm{Z}_{1,i} \\ \bm{Z}_{2,ij} \end{bmatrix} \begin{bmatrix} \bm{Z}_{1,i} \\ \bm{Z}_{2,ij} \end{bmatrix}^\top \right\rangle 
 \right)^{-1}.
\end{align} 
The measurement error variance $\sigma^2$ is updated by computing the expected reconstruction error across all observations:
\begin{equation}\label{eq:sigma_update}
    \widetilde \sigma^2 = \frac{1}{p\sum_{i=1}^n J_i}\sum_{i = 1}^n \sum_{j \in \mathcal{O}_i}\langle \|\bm{Y}_{ij} - \widetilde{\bm{ W}}_1 \bm{Z}_{1, i} - \widetilde{\bm{ W}}_2 \bm{Z}_{2, ij}\|_2^2 \rangle.
\end{equation}
We take a Newton-Raphson step utilizing the exact gradient and Hessian with respect to the kernel length-scale parameters $\ell^{(r)}$. The proofs for the M-step updates and matrix calculus for the kernel parameters are provided in Web Appendix A. 

We iterate between the E-step and M-step until the parameter estimates converge. Let the maximum relative error between two matrices $\bm{M}_1$ and $\bm{M}_2$ be defined as:
\[
\text{maxRelErr}(\bm{M}_1, \bm{M}_2) = \max_{k, r} \frac{|\bm{M}_{1,kr} - \bm{M}_{2, kr}|}{|\bm{M}_{2, kr}| + \epsilon},
\]
where $\epsilon = 10^{-12}$ is a small constant added to prevent division by zero. The algorithm is deemed to converge when the maximum relative changes in $\bm{W}_1\bm{W}_1^\top$, $\bm{W}_2\bm{W}_2^\top$, the noise variance $\sigma^2$ and the kernel parameters (if applicable) all fall below a predefined tolerance level. The complete EM algorithm for estimating the HPPCA model is summarized in Algorithm \ref{alg:hppca-em}.

Computationally, the E-step requires solving an independent $(d_1+d_2J_i)\times(d_1+d_2J_i)$ linear system for the $i$th subject, whereas the loading update in the M-step involves only a global $(d_1+d_2)\times(d_1+d_2)$ system. Since $d_1$ and $d_2$ are much smaller than $p$ and $J_i$ is typically modest, the resulting computations remain tractable and can be parallelized across subjects in the E-step. 

\begin{algorithm}[!htbp]
\caption{\textbf{EM Algorithm for HPPCA}}
\label{alg:hppca-em}
\begin{algorithmic}[1]
\State \textbf{Input:} Centered responses $\{\bm{Y}_{ij}\}$, dimensions $(d_1,d_2)$, kernel type, survey times $\{t_{ij}\}$, tolerance \texttt{tol}, max iterations \texttt{maxit}.
\State \textbf{Output:} Loadings $(\widehat{\bm{W}}_1, \widehat{\bm{W}}_2)$, noise $\hat{\sigma}^2$, length-scales $\{\hat{\ell}^{(r)}\}$, posteriors $\{\widehat{\bm{Z}}_{1,i}, \widehat{\bm{Z}}_{2,ij}\}$, and imputations $\widehat{\bm{Y}}_{ij}^{(m)}$.
\State \textbf{Initialize:} $\bm{W}_1, \bm{W}_2, \sigma^2, \{\ell^{(r)}\}$; set $iter \leftarrow 0$, $eps \leftarrow \infty$.
\While{$eps \ge \texttt{tol}$ \textbf{and} $iter < \texttt{maxit}$}
\State \textbf{E-step} (Parallel for $i=1,\ldots,n$): 
\State \quad Construct posterior covariance $\bm{\Sigma}_{Z_i}$ (Eq. \ref{eq:post_var}) and mean $\bm{\mu}_{Z_i}$ (Eq. \ref{eq:post_mean}).
\State \quad Calculate expected sufficient statistics $\langle \bm{Z}_i \rangle, \langle \bm{Z}_i \bm{Z}_i^\top \rangle, \langle \bm{Y}_{ij}^{(m)} \rangle, \langle \bm{Y}_{ij}^{(m)}\bm{Y}_{ij}^{(m)\top} \rangle, \langle \bm{Y}_{ij}^{(m)}\bm{Z}_{1,i}^\top \rangle$, $ \langle \bm{Y}_{ij}^{(m)}\bm{Z}_{2,ij}^{\top} \rangle$ and  matrices $\widetilde{\bm{G}}_{ij}^{(1)}, \widetilde{\bm{G}}_{ij}^{(2)}$.
\State \textbf{M-step:} 
\State \quad Update $\widetilde{\bm{W}}_1, \widetilde{\bm{W}}_2$ by solving the augmented normal equations (Eq. \ref{eq:W_update}).
\State \quad Update noise variance $\widetilde{\sigma}^2$ (Eq. \ref{eq:sigma_update}).
\State \quad \textbf{Kernel update (if applicable):} For $r=1,\ldots,d_2$, take a Newton-Raphson step to update $\ell^{(r)}$ using analytical gradients and Hessians.
\State \textbf{Convergence:} 
\State \quad $eps \leftarrow \max \!\Big\{\text{maxRelErr}(\bm{W}_1\bm{W}_1^\top, \widetilde{\bm{W}}_1\widetilde{\bm{W}}_1^\top),\ \text{maxRelErr}(\bm{W}_2\bm{W}_2^\top, \widetilde{\bm{W}}_2\widetilde{\bm{W}}_2^\top),$
\State \quad \quad \quad \quad \quad \quad \ \ $\frac{|\widetilde{\sigma}^2-\sigma^2|}{\sigma^2+\epsilon},\ \max_{r} \frac{|\widetilde{\ell}^{(r)}-\ell^{(r)}|}{\ell^{(r)}+\epsilon} \Big\}$.
\State \quad $(\bm{W}_1, \bm{W}_2, \sigma^2, \ell^{(r)}) \leftarrow (\widetilde{\bm{W}}_1, \widetilde{\bm{W}}_2, \widetilde{\sigma}^2, \widetilde{\ell}^{(r)})$; \quad $iter \leftarrow iter + 1$.
\EndWhile
\end{algorithmic}
\end{algorithm}

\subsection{Special-Case Initializers under Shared Temporal Covariance}
\label{sec:init-cases}

The EM algorithm converges more rapidly when it is initialized with values near the final solution. To obtain computationally efficient deterministic starting values and intuition of the HPPCA loadings, we consider a balanced surrogate problem in which subjects are observed on a common grid of $J\ge 2$ visits without missingness and the within-subject latent dimensions share a common temporal covariance matrix $\bm{\Sigma}$. The resulting estimators are used to initialize the general EM algorithm and also provide intuition for how HPPCA separates the between-subject and within-subject covariance structures.

Under this surrogate model, 
\[
\bm y_i\stackrel{\text{i.i.d.}}{\sim}
\mathcal N\!\left(
\bm 0,\,
\one_J\one_J^\top \otimes \A
+
\bm{\Sigma}\otimes \B
+
\sigma^2 \I_{pJ}
\right),
\qquad i=1,\ldots,n.
\]

\begin{assumption}[Shared temporal eigenspace]
\label{ass:ERS}
The temporal covariance matrix $\bm{\Sigma}$ is symmetric positive semidefinite and satisfies $\bm{\Sigma}\one_J=\lambda_1\one_J$ for some $\lambda_1\ge 0$. Equivalently, there exists an orthogonal matrix $\bm U_T=[\bm u_1,\bm U_\perp]$, where $\bm u_1=\one_J/\sqrt{J}$, such that $\bm U_T^\top \bm{\Sigma}\bm U_T=\diag(\lambda_1,\lambda_2,\ldots,\lambda_J)$, and $\bm U_T^\top(\one_J\one_J^\top)\bm U_T=\diag(J,0,\ldots,0)$.
\end{assumption}

Define the rotated time features $\widetilde{\bm Y}_i=\bm Y_i\bm U_T=[\widetilde{\bm Y}_{i1},\ldots,\widetilde{\bm Y}_{iJ}]$, together with
$\bm S_t=n^{-1}\sum_{i=1}^n \widetilde{\bm Y}_{it}\widetilde{\bm Y}_{it}^\top$ ($t=1,\ldots,J$), and $\overline{\bm S}_c=(J-1)^{-1}\sum_{t=2}^J \bm S_t$.

\begin{lemma}[Likelihood decomposition]
\label{lem:ERS-deco}
Under Assumption~\ref{ass:ERS}, the rotated columns are mutually independent, with
\[
\widetilde{\bm Y}_{i1}\sim\mathcal N\!\left(\bm 0,\,
J\A+\lambda_1\B+\sigma^2 \I_p\right),
\]
and
\[
\widetilde{\bm Y}_{it}\sim\mathcal N\!\left(\bm 0,\,
\lambda_t\B+\sigma^2 \I_p\right),
\qquad t=2,\ldots,J.
\]
Consequently, up to an additive constant, the negative log-likelihood decomposes as
\[
\mathcal L(\A,\B,\sigma^2)
=
\mathcal L_A(\A\mid \B,\sigma^2)+\mathcal L_B(\B,\sigma^2),
\]
where
\begin{align}
\mathcal L_A
&=
\frac{n}{2}\Big[
\log\det(J\A+\lambda_1\B+\sigma^2 \I_p)
+
\tr\!\big\{\bm S_1(J\A+\lambda_1\B+\sigma^2 \I_p)^{-1}\big\}
\Big],
\label{eq:Apart-log}
\\
\mathcal L_B
&=
\frac{n}{2}\sum_{t=2}^J
\Big[
\log\det(\lambda_t\B+\sigma^2 \I_p)
+
\tr\!\big\{\bm S_t(\lambda_t\B+\sigma^2 \I_p)^{-1}\big\}
\Big].
\label{eq:Bpart-log}
\end{align}
\end{lemma}

Lemma~\ref{lem:ERS-deco} shows that the rotated contrast directions depend only on $(\B,\sigma^2)$, whereas the rotated mean direction contains the additional between-subject covariance component $\A$. Thus, once $(\B,\sigma^2)$ are estimated from the contrast directions, $\A$ can be recovered from a profiled low-rank covariance problem.

\begin{theorem}[Profile update for $\A$]
\label{thm:A-closed-form}
For fixed $(\B,\sigma^2)$ with $\sigma^2>0$, let
\[
\bm H=\lambda_1\B+\sigma^2 \I_p.
\]
Over the parameter space
\[
\A\succeq 0,
\qquad
\rank(\A)\le d_1,
\]
a minimizer of $\mathcal L_A$ is
\[
\A^\star
=
\frac{1}{J}\bm H^{1/2}
\Pi_+^{(d_1)}
\!\left(
\bm H^{-1/2}\bm S_1\bm H^{-1/2}-\I_p
\right)
\bm H^{1/2},
\]
where $\Pi_+^{(d_1)}(\cdot)$ denotes the spectral projection that retains the largest $d_1$ positive eigenvalues and sets the remaining eigenvalues to zero.
\end{theorem}

For general known shared temporal covariance satisfying Assumption~\ref{ass:ERS}, we estimate $(\B,\sigma^2)$ from the contrast likelihood \eqref{eq:Bpart-log} by block-coordinate descent in the spectral parametrization $\B=\bm{Q}\diag(\bm{b})\bm{Q}^\top$, alternating a projected IRLS update for $(\bm{b},\sigma^2)$ and a Riemannian gradient step for $\bm{Q}$ on the Stiefel manifold. This yields \emph{Initializer~1}; details are given in Web Algorithm~1.

A particularly simple case arises under compound symmetry.

\begin{corollary}[Compound-symmetry special case]
\label{cor:cs}
Suppose that
\[
\bm{\Sigma}(\tau^2)=(1-\tau^2)I_J+\tau^2\one_J\one_J^\top,
\qquad 0\le \tau^2<1.
\]
Then
\[
\lambda_1(\tau^2)=1+(J-1)\tau^2,
\qquad
\lambda_t(\tau^2)\equiv \lambda_c(\tau^2)=1-\tau^2,
\quad t=2,\ldots,J.
\]
Let
\[
\overline{\bm S}_c=\Q\diag(\bar s_1,\ldots,\bar s_p)\Q^\top,
\qquad
\bar s_1\ge\cdots\ge\bar s_p,
\]
and assume $d_2<p$. For fixed $\tau^2$, the contrast likelihood \eqref{eq:Bpart-log} is minimized by
\[
\widehat{\sigma}^2=\frac{1}{p-d_2}\sum_{k>d_2}\bar s_k,
\qquad
\widehat b_k=
\begin{cases}
\max\!\left\{
\frac{\bar s_k-\widehat{\sigma}^2}{\lambda_c(\tau^2)},\,0
\right\}, & k=1,\ldots,d_2,\\[2mm]
0, & k=d_2+1,\ldots,p.
\end{cases}
\]
with $\widehat{\B}=\bm{Q}\diag(\widehat b_1,\ldots,\widehat b_p)\bm{Q}^\top$.
\end{corollary}

Combined with Theorem~\ref{thm:A-closed-form}, Corollary~\ref{cor:cs} yields \emph{Initializer~2}: an equal-contrast PPCA step followed by a one-dimensional bounded profile search over $\tau^2$ and the profiled update of $\A$, as detailed in Web Algorithm~2. Full derivations, proofs, and algorithms are also provided in Web Appendix B. Both initialization methods are implemented and available as user-selectable options in our accompanying Python code.

\section{Simulation Studies}\label{sec:simulation}

We conducted extensive simulation studies to evaluate the performance of the proposed HPPCA framework. The experiments were designed to systematically assess (1) the accuracy of parameter recovery under heavy missingness; (2) imputation performance compared to existing dimension reduction methods under correct model specification; (3) the robustness of the model when the true latent dynamics are misspecified; and (4) the computational efficiency gained by the balanced-panel initializers proposed in Section \ref{sec:init-cases}.

Across all scenarios, we fixed the sample size to $n = 1000$ and the number of features to $p = 50$. We also assumed that the temporal kernels share the same kernel parameter $\ell$ through all the simulations. We applied item-level missing completely-at-random (MCAR) with missing rates $p_{\text{miss}}$ ranging from $0.1$ to $0.9$. We benchmarked HPPCA against standard PPCA, which ignores longitudinal dependency by treating repeated measures from the same subject as independent observations, and two variants of mFPCA developed for time-series data. The two mFPCA approaches are distinguished by their eigen-decomposition strategies:  mFPCA(cov) relies on the classical diagonalization of the covariance operator \citep{Happ2018multivariate}, and mFPCA(inner) decomposes the inner-product matrix to compute the eigencomponents \citep{golovkine2025use}. PPCA and mFPCA were implemented via the python libraries \textit{pyppca} \citep{green2019pyppca} and \textit{FDApy} \citep{golovkine_2024_fdapy_paper}, respectively.

To ensure a rigorous and conservative evaluation, we matched the visit-level dimensions across methods. Specifically, when HPPCA is configured with latent dimensions $d_1 = d_2 = d$, we allocate $d_1 + d_2 = 2d$ principal components to all PPCA and mFPCA variants. It is worth noting that this configuration inherently favors the existing methods. Because HPPCA structurally constrains the $d_1$-dimensional static traits to be shared across a subject's longitudinal trajectory, its total latent dimension for the $i$th subject is restricted to $d_1 + J_i d_2$. In contrast, the baselines are granted a larger, unconstrained latent capacity of $J_i (d_1 + d_2)$-dimensions per subject to model the observations.

Because the baseline mFPCA methods require sufficient temporal density for functional smoothing, they encounter numerical instability at extreme missing rates. Furthermore, extreme missingness rates may lack practical plausibility in certain real-world scenarios. Consequently, our predictive imputation comparisons against baselines are restricted to $p_{\text{miss}} \in \{0.1, 0.3, 0.5\}$. All results are averaged over 100 independent replicates.

\subsection{Parameter Recovery under Heavy Missingness}\label{subsec:study1}
In this simulation, we adopt three kernels for $\bm{\Sigma}$: (i) i.i.d. ($\bm{\Sigma}=\I_J$), (ii) RBF$(K_{RBF}(\ell=10))$, and (iii) Mat\'ern–$5/2$ $(K_{5/2}(\ell=10))$. We fixed the noise component $\sigma^2$ at 0.25 and evaluated panels with varying temporal lengths and latent dimensions: $(J, d_1, d_2) \in \{(5, 2, 2), (10, 4, 4)\}$. For simplicity, we assumed a standardized follow-up schedule where all subjects share the same set of observation times, spaced by a fixed interval of 10. The entries of the true loading matrices $\bm{W}_1$ and  $\bm{W}_2$ were generated independently from a standard normal distribution. To assess factor loadings recovery, we compared the estimated and true column spaces via principal angles. Let $\mathcal M(\bm{W}_a)$ denote the span of the columns of $\bm{W}_a$, and let $\theta_a \in [0^\circ, 90^\circ]$ be the largest principal angle between $\mathcal M(\widehat{\bm{W}}_a)$ and $\mathcal M(\bm{W}_a)$, $a \in\{1,2\}$, whose cosine can be obtained from the smallest singular values of the cross–basis matrix \citep{bjorck1973numerical, absil2006largest}. A value closer to $0^\circ$ means the two subspaces are closer.

Table \ref{tab:sim1-para} demonstrates that the proposed EM algorithm successfully recovers the hierarchical structure. Subspace errors increase gracefully with the missing rate but remain exceptionally small, typically below $4^\circ$ when missing rate is 90\%. In addition, $\sigma^2$ and $\ell$ are accurately recovered without bias. 

As the missing rate reaches $0.7$ and $0.9$, the data design degenerates into a highly irregular structure where many subjects missed a whole visit. While existing methods may not be stable under such sparsity, HPPCA natively accommodates this irregularity via its exact probabilistic missing-data mechanism, successfully leveraging the remaining fragments of temporal correlation.

\begin{table}[!htbp]
\begin{threeparttable}
  \centering
  \caption{Subspace Recovery and Parameter Estimation for HPPCA with Different Temporal Covariance Structures and Missing Rates over 100 Replicates.}
  \footnotesize
    \begin{tabular}{llccccc}
    \toprule
    $(J, d_1, d_2)$ & Kernel & $p_{\text{miss}}$ & $\bm{W}_1$    & $\bm{W}_2$    & $\sigma^2$ & $\ell$ \\
    \midrule
    (5, 2, 2) & IID   & 0     & 1.574 (0.609) & 0.465 (0.054) & 0.000 (0.001) & / \\
          &       & 0.1   & 1.581 (0.610) & 0.495 (0.059) & 0.000 (0.001) & / \\
          &       & 0.3   & 1.597 (0.606) & 0.568 (0.069) & 0.000 (0.001) & / \\
          &       & 0.5   & 1.647 (0.598) & 0.680 (0.085) & -0.001 (0.001) & / \\
          &       & 0.7   & 1.742 (0.586) & 0.910 (0.114) & -0.001 (0.001) & / \\
          &       & 0.9   & 3.334 (4.168) & 6.206 (18.205) & 0.002 (0.023) & / \\
          & $K_{RBF}(\ell = 10)$ & 0     & 2.025 (0.799) & 0.471 (0.063) & 0.000 (0.001) & 0.012 (0.077) \\
          &       & 0.1   & 2.034 (0.797) & 0.501 (0.068) & 0.000 (0.001) & 0.013 (0.078) \\
          &       & 0.3   & 2.060 (0.798) & 0.566 (0.076) & 0.000 (0.001) & 0.013 (0.078) \\
          &       & 0.5   & 2.093 (0.776) & 0.681 (0.086) & -0.001 (0.001) & 0.013 (0.079) \\
          &       & 0.7   & 2.159 (0.740) & 0.899 (0.120) & -0.001 (0.001) & 0.015 (0.084) \\
          &       & 0.9   & 2.805 (0.626) & 2.014 (0.267) & -0.003 (0.003) & 0.014 (0.117) \\
          & $K_{5/2}(\ell = 10)$  & 0     & 1.930 (0.735) & 0.475 (0.066) & 0.000 (0.001) & 0.026 (0.126) \\
          &       & 0.1   & 1.940 (0.732) & 0.505 (0.066) & 0.000 (0.001) & 0.026 (0.126) \\
          &       & 0.3   & 1.966 (0.735) & 0.569 (0.074) & 0.000 (0.001) & 0.025 (0.127) \\
          &       & 0.5   & 2.009 (0.713) & 0.683 (0.094) & -0.001 (0.001) & 0.024 (0.129) \\
          &       & 0.7   & 2.088 (0.695) & 0.899 (0.121) & -0.001 (0.001) & 0.026 (0.127) \\
          &       & 0.9   & 2.772 (0.605) & 2.017 (0.282) & -0.003 (0.003) & 0.036 (0.172) \\
    (10, 4, 4) & IID   & 0     & 1.848 (0.450) & 0.375 (0.048) & 0.000 (0.001) & / \\
          &       & 0.1   & 1.851 (0.448) & 0.398 (0.051) & 0.000 (0.001) & / \\
          &       & 0.3   & 1.858 (0.437) & 0.458 (0.057) & 0.000 (0.001) & / \\
          &       & 0.5   & 1.880 (0.431) & 0.560 (0.067) & 0.000 (0.001) & / \\
          &       & 0.7   & 1.952 (0.421) & 0.792 (0.095) & -0.001 (0.001) & / \\
          &       & 0.9   & 2.832 (0.428) & 2.275 (0.273) & -0.005 (0.003) & / \\
          &  $K_{RBF}(\ell = 10)$ & 0     & 2.712 (0.568) & 0.378 (0.043) & 0.000 (0.001) & 0.003 (0.028) \\
          &       & 0.1   & 2.719 (0.568) & 0.400 (0.046) & 0.000 (0.001) & 0.003 (0.028) \\
          &       & 0.3   & 2.724 (0.568) & 0.460 (0.050) & 0.000 (0.001) & 0.002 (0.029) \\
          &       & 0.5   & 2.737 (0.574) & 0.562 (0.060) & 0.000 (0.001) & 0.003 (0.030) \\
          &       & 0.7   & 2.784 (0.572) & 0.782 (0.092) & -0.001 (0.001) & 0.001 (0.033) \\
          &       & 0.9   & 3.352 (0.577) & 2.081 (0.246) & -0.005 (0.003) & 0.001 (0.060) \\
          &  $K_{5/2}(\ell = 10)$ & 0     & 2.667 (0.575) & 0.377 (0.043) & 0.000 (0.001) & 0.003 (0.055) \\
          &       & 0.1   & 2.673 (0.574) & 0.399 (0.047) & 0.000 (0.001) & 0.003 (0.055) \\
          &       & 0.3   & 2.680 (0.573) & 0.459 (0.051) & 0.000 (0.001) & 0.003 (0.056) \\
          &       & 0.5   & 2.693 (0.577) & 0.563 (0.062) & 0.000 (0.001) & 0.003 (0.057) \\
          &       & 0.7   & 2.743 (0.570) & 0.786 (0.094) & -0.001 (0.001) & 0.002 (0.058) \\
          &       & 0.9   & 3.357 (0.573) & 2.176 (0.262) & -0.005 (0.003) & 0.002 (0.086) \\
    \bottomrule
    \end{tabular}%
    Note:  Entries in the \(\mathbf W_1\) and \(\mathbf W_2\) columns are the largest principal
angles (in degrees) between the estimated and ground–truth loading subspaces,
reported as mean (sd). The \(\sigma^2\) and \(\ell\) columns give the
bias (sd) of the estimated noise variance and kernel length–scale,
respectively.
  \label{tab:sim1-para}%
\end{threeparttable}%
\end{table}%

\subsection{Imputation Accuracy under Correct Specification}

A key advantage of our generative framework is its ability to optimally impute unobserved clinical features via the conditional posterior mean $\langle \bm{Y}_{ij}^{(m)} \rangle$. Using the correctly specified datasets, we evaluated imputation accuracy via the mean squared error (MSE) computed on the randomly masked, held-out entries: $\text{MSE}_{\text{miss}} = \sum_i (\hat{y}_i - y_i)^2 / \#\text{missing}$.

When simulating correctly specified datasets, we spanned none to strong temporal correlation by employing an uncorrelated i.i.d. case $\bm{\Sigma}=\mathbf I_J$ and an RBF kernel with length-scale $\ell = 10$. Similar to the first simulation study, we set $\sigma^2=0.25$, and considered two settings $(J, d_1,d_2)\in\{(5, 2,2),(10, 4,4)\}$ and a fixed time interval of $10$ between visits. The entries of the true loading matrices $\bm{W}_1$ and $\bm{W}_2$ were also generated independently from a standard normal distribution.  

Figure \ref{fig:sim_correct} displays the distribution of the imputation MSE. HPPCA consistently outperforms PPCA and both mFPCA variants across all missing rates, visit lengths, and temporal kernels. PPCA performs worse than HPPCA because it fails to borrow information across time. Therefore, its performance deficit relative to HPPCA widens as the number of scheduled visits and the missingness rate increase. Conversely, while mFPCA variants attempt to smooth over time, they struggle with sparse observed timepoints. Therefore, their imputation MSE is much higher than that of HPPCA and PPCA. 

\begin{figure}[!htbp]
    \centering
    \includegraphics[width=0.9\linewidth]{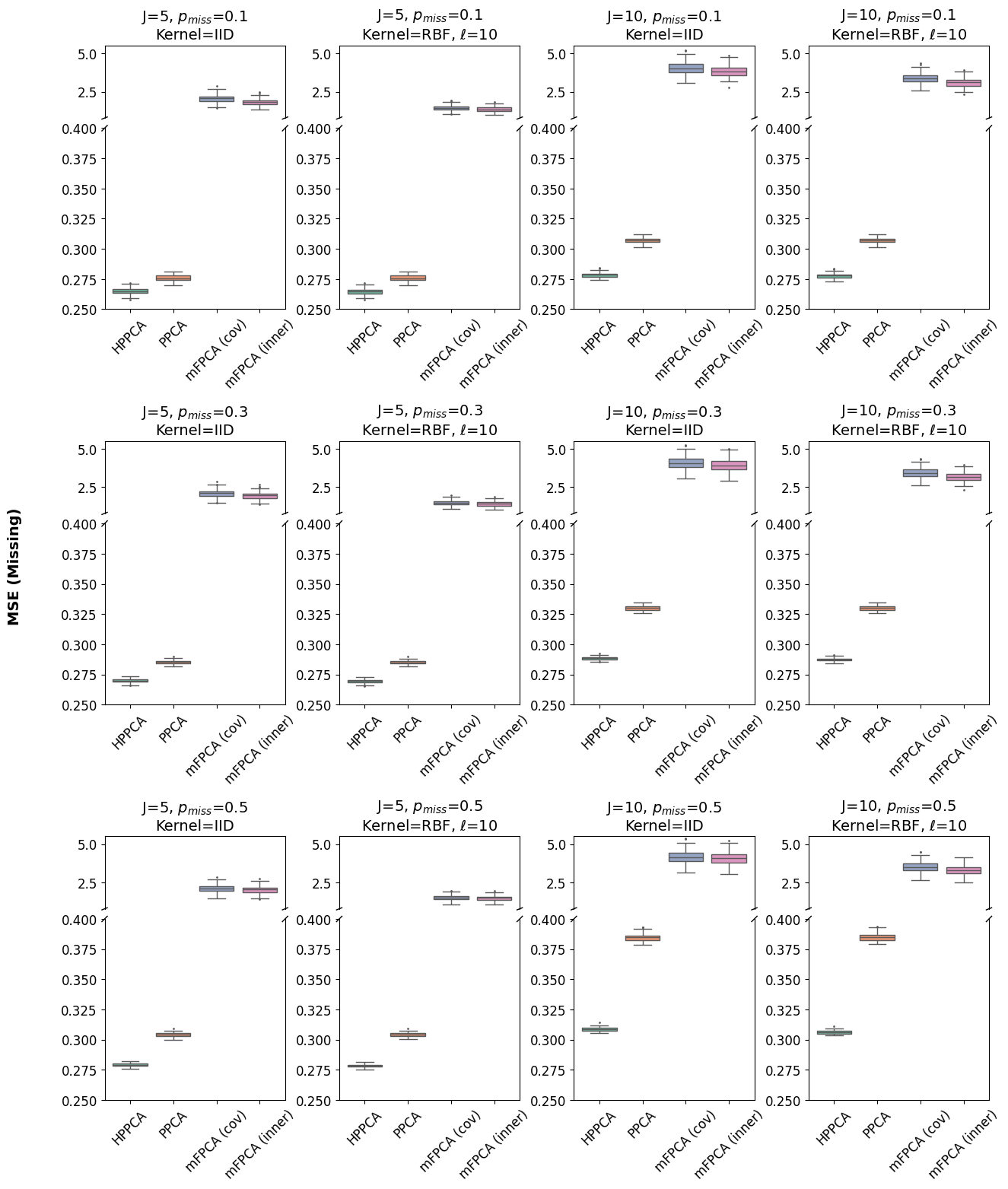}
    \caption{\textbf{Imputation Accuracy under Correct Model Specification.} Distribution of the imputation MSE on held-out entries across 100 simulated replicates. In each panel, from left to right, the green, orange, blue, pink boxplots pertain to HPPCA, PPCA, mFPCA(cov) and mFPCA(inner). Rows correspond to different entry-wise missingness rates ($p_{\text{miss}} \in \{0.1, 0.3, 0.5\}$). Columns denote different combinations of scheduled visits ($J \in \{5, 10\}$) and true within-subject temporal correlation structures (uncorrelated IID vs. smooth RBF kernel with $\ell=10$). A broken y-axis is employed to accommodate the substantially larger errors produced by the mFPCA variants while preserving the visual resolution needed to differentiate the superior performance of HPPCA against PPCA.}
    \label{fig:sim_correct}
\end{figure}

\subsection{Robustness to Model Misspecification}
While the HPPCA model assumes that within-subject latent trajectories evolve as smooth, continuous-time Gaussian processes, real-world subject trajectories often exhibit rough, discrete state transitions. To rigorously stress-test our model's robustness against this assumption, we generated the true latent dynamics using a discrete-time linear dynamical system (LDS) \citep{bishop2006pattern}.  We evaluated two specific experimental settings: $(J, d) \in \{(5, 4), (10, 8)\}$, where $J$ denotes the number of visits per subject and $d$ represents the total number of latent dimensions in the LDS.

Specifically, the latent factors $\bm{Z}_{ij}$ for the $i$th subject at the $j$th time step were generated via a first-order autoregressive (AR(1)) process: $ \bm{Z}_{ij} = \bm{\Phi} \bm{Z}_{i(j-1)} + \bm{w}_{ij}$. To mimic both slowly varying subject-level traits and variation with lower temporal persistence, we partitioned the $d$-dimensional latent space into two subspaces of dimensions $d_1$ and $d_2$ (such that $d = d_1 + d_2$). The state transition matrix was defined as $\bm{\Phi} = \operatorname{blockdiag} \{0.95 \bm{I}_{d_1}, \rho \bm{I}_{d_2}\}$. The process noise $\bm{w}_{ij}$ was drawn from $\mathcal{N}(\bm{0}, \bm{\Omega}_w)$, with $\bm{\Omega}_w = \operatorname{blockdiag}((1-0.95^2)\bm{I}_{d_1}, (1-\rho^2)\bm{I}_{d_2})$. The parameter $\rho \in \{0.3, 0.6, 0.95\}$ governs the temporal autocorrelation of the dynamic factors, while the fixed $0.95$ coefficient ensures the first $d_1$ dimensions maintain a consistently high correlation. To ensure that the autoregressive process is covariance-stationary, the initial latent state for each subject was drawn from a standard multivariate normal distribution, specifically $\bm{Z}_{i0} \sim \mathcal{N}(\bm{0}, \bm{I}_d)$. Finally, the latent representation were related to the observed features in a PPCA structure, i.e. $\bm{Y}_{ij} = \bm{W}\bm{Z}_{ij} + \sigma \bm{\varepsilon}_{ij}, \bm{\varepsilon}_{ij} \sim N(\bm{0}, \bm{I}_p)$, where the entries of the loading matrix $\bm{W} \in \mathbb{R}^{p \times d}$ were generated from a standard normal distribution, and $\sigma^2$ was fixed at 0.25. In this simulation, all observation times were uniformly scaled to the interval $[0, 1]$ and evenly distributed. 

Figure \ref{fig:sim_misspecified_ar1_J5} presents the imputation MSE under this misspecified scenario for $J = 5$ and $d = 4$. Analogous results for $J = 10$ and $d = 8$, along with similarly robust findings for AR(2) generating processes, are detailed in Web Appendix~C.1 and C.2 respectively. Despite the generative mismatch, HPPCA with RBF kernel successfully adapts and retains its predictive dominance over all baselines. The GP kernels gracefully approximate the discrete AR dynamics by optimizing the data-driven length-scale $\ell$. When the missing rate become higher, the gaps between the existing methods and HPPCA become bigger. PPCA's performance degrades when the time autocorrelation becomes stronger, while the mFPCA(cov)'s performance degrades more when the time autocorrelation becomes weaker. 
These results confirm HPPCA's flexibility as a general longitudinal imputation tool.

\begin{figure}[!htbp]
    \centering
\includegraphics[width=0.8\linewidth]{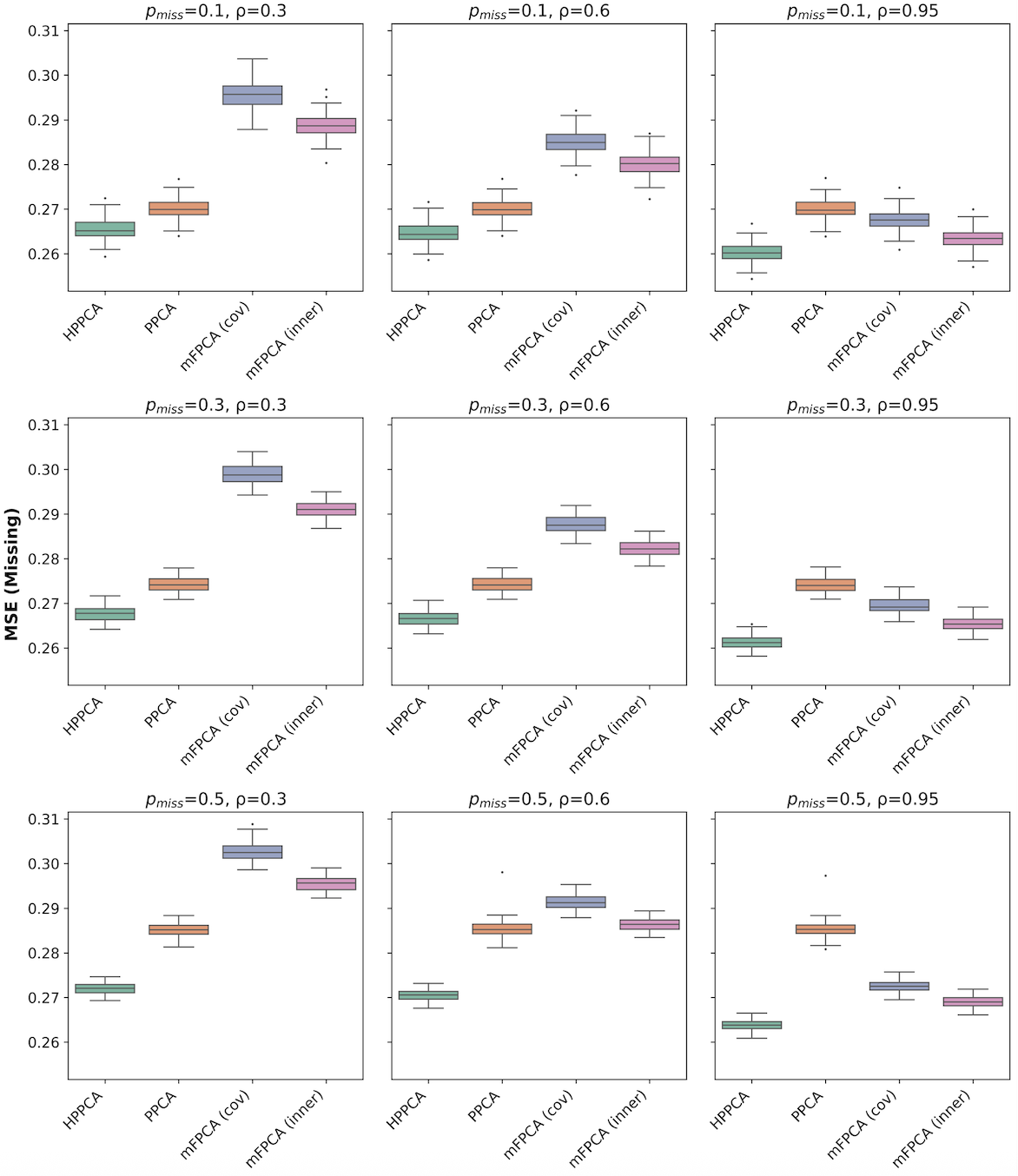}
    \caption{\textbf{Predictive performance on missing observations under the discrete LDS generative mechanism ($J=5, d=4$).} Box plots display the MSE of missing data reconstruction for HPPCA and three existing methods (PPCA, mFPCA-cov, mFPCA-inner). The experimental grid varies the missingness probability $p_{\text{miss}} \in \{0.1, 0.3, 0.5\}$ (rows) and the temporal autocorrelation parameter $\rho \in \{0.3, 0.6, 0.95\}$ (columns).}
    \label{fig:sim_misspecified_ar1_J5}
\end{figure}

\subsection{Computational Efficiency of the Initializers}
 To demonstrate the practical utility and computational efficiency of the initializers derived in Section \ref{sec:init-cases}, we replicated the simulation settings from Section~\ref{subsec:study1} using an RBF kernel with $\ell = 10$ and missing rates $p_{\text{miss}} \in \{0, 0.1, 0.3, 0.5\}$. Over 100 replicates for each scenario, we compared our two initializers with a random initializer where the entries of $\bm{W}_1$ and $\bm{W}_2$ are sampled independently from a standard normal distribution, and $\sigma^2 = 1.0$, and the length scale set to $\ell = \max(5, T)$, where $T$ represents the total time range. Web Appendix~C.3 details the implementation of Initializers 1 and 2 in real scenarios characterized by unknown temporal covariance matrix, irregular visit schedules, and heavy missingness.
 
 Web Table~1 and~2 present the comprehensive comparison of the three initialization methods for $(J, d_1, d_2) = (5, 2, 2)$ and $(10, 4, 4)$, respectively. Across all settings, the EM algorithm demonstrates robust final convergence, achieving comparable, high-quality parameter estimates regardless of the initialization method.
 
 A key advantage of the proposed initializers is their exceptional proximity to the true latent subspaces. When initialized randomly from standard normal distributions, the starting loading matrices are virtually orthogonal to the true subspaces (initial principal angles $\approx 85^\circ \sim 87^\circ$). By contrast, despite the structural misspecification of the surrogate model and the presence of mean-imputed missing data, both Initializer 1 and Initializer 2 provide remarkable warm starts when the panel is short, even under 50\% of missingness. Their initial principal angles are remarkably closer to the truth (e.g., $2^\circ \sim 7^\circ$). When the panel becomes longer, initializer 2 still has robust performance across different missing rates. Initializer 1 occasionally shows a larger initial discrepancy than Initializer 2, but it still vastly outperforms random starts.

These high-quality starting points translate into massive computational savings during the subsequent EM algorithm updates. For example, in the complex scenario of $(J, d_1, d_2) = (10, 4, 4)$ with a $30\%$ missingness rate, the random initializer required an average of $11,489$ iterations (185.3 minutes) to reach convergence. Initializer 1 reduced the EM iteration count by over 55\% (down to 5,130 iterations), saving more than 100 minutes of wall-clock time per run.  Notably, even when Initializer 1 starts with a larger subspace discrepancy than Initializer 2, its total wall-clock time shrinks to be nearly identical. This rapid alignment suggests that the EM algorithm effectively corrects the Initializer 1 subspace during the first few gradient steps. These results highlight the practical effectiveness of the proposed initializers for robustly scaling the HPPCA framework to complex datasets.

\section{Application to the RECOVER Adult Cohort}\label{sec:recover}

The NIH RECOVER initiative includes a nationwide longitudinal adult cohort study designed to understand the post-acute sequelae of SARS-CoV-2 infection, commonly known as long COVID \citep{horwitz2023researching}. The physical symptom burden of long COVID is highly heterogeneous, exhibiting complex time dynamics that evolve considerably from the initial acute infection through the post-acute phase. Although the scheduled follow-up is every three months, clinical visits are inherently irregular and subject to heavy item-level missingness. Disentangling these static and dynamic components thus requires a flexible hierarchical framework.

\subsection{Study Cohort and Data Preprocessing}
We analyzed the longitudinal physical symptom severity trajectories of subjects from the RECOVER adult cohort, utilizing data updated through June 2025. To focus strictly on post-infection trajectories, the analysis cohort was restricted to $n = 13{,}192$ subjects with at least one documented SARS-CoV-2 infection. This comprises $7{,}873$ subjects enrolled in the post-acute phase, $4{,}212$ in the acute phase, and $1{,}107$ who were uninfected at enrollment but subsequently infected. For the $i$th subject, the observation time $t_{ij}$ at the $j$th visit was defined as the number of days since their first infection date to COVID-19. 

The clinical feature space consists of $p=57$ non-sex-specific physical symptom burden items. Each item was originally recorded on an ordinal 5-category severity scale (1: ``Not at all'' to 5: ``Very Much''). To efficiently handle logical skip patterns in the surveys, we applied an imputation: if a subject explicitly reported the absence of the symptom at a visit, the corresponding detailed burden questions were filled with ``Not at all'' (score 1). Prior to analysis, each of the $57$ symptom variables was standardized to have zero mean and unit variance.

Despite this logical filling, the dataset remains highly sparse, with an overall item-level missingness rate of $28.58\%$. As presented in Web Figure~3, the number of longitudinal visits varies substantially across subjects, ranging from 1 to 16 visits; the time ranges from the first day of COVID-19 infection to more than 1400 days, which makes it well suited for studying long COVID. The empirical distributions of visit counts and timings, along with detailed physical symptom preprocessing steps and definitions, are provided in Web Appendix D.1 and D.2, respectively.

\subsection{Extracting Static and Dynamic Symptom Phenotypes}
We fitted the proposed HPPCA model using an RBF temporal kernel with a common kernel parameter $\ell$, evaluating matched latent dimensions $d_1 = d_2 \in \{1,\ldots,10\}$. As model capacity increases, both the estimated measurement noise $\sigma^2$ and the reconstruction MSE monotonically decline (See Web Figure~4), which shows that the HPPCA model captures more variation when both level dimensions increase. 
To quantitatively verify that the extracted latent factors capture diverse clinical constructs within each hierarchical level, we computed the pairwise angles between the columns of $\widehat{\bm{W}}_1$, as well as between the columns of $\widehat{\bm{W}}_2$, for $d_1 = d_2 = 10$. These intra-matrix angles predominantly range between 55° and 90° (See Web Figure 5), indicating that the individual latent dimensions represent distinct symptom profiles rather than highly correlated, redundant signals.

A major advantage of HPPCA over traditional methods is its ability to learn separate linear subspaces for time-invariant, between-subject static traits ($\bm{W}_1$) and time-varying, within-subject dynamic fluctuations ($\bm{W}_2$). Figure~\ref{fig:recover_heatmaps} presents one set of estimated factor loadings at $d_1 = d_2 = 5$.

\begin{figure}[!htbp]
    \centering
    \includegraphics[width=\linewidth]{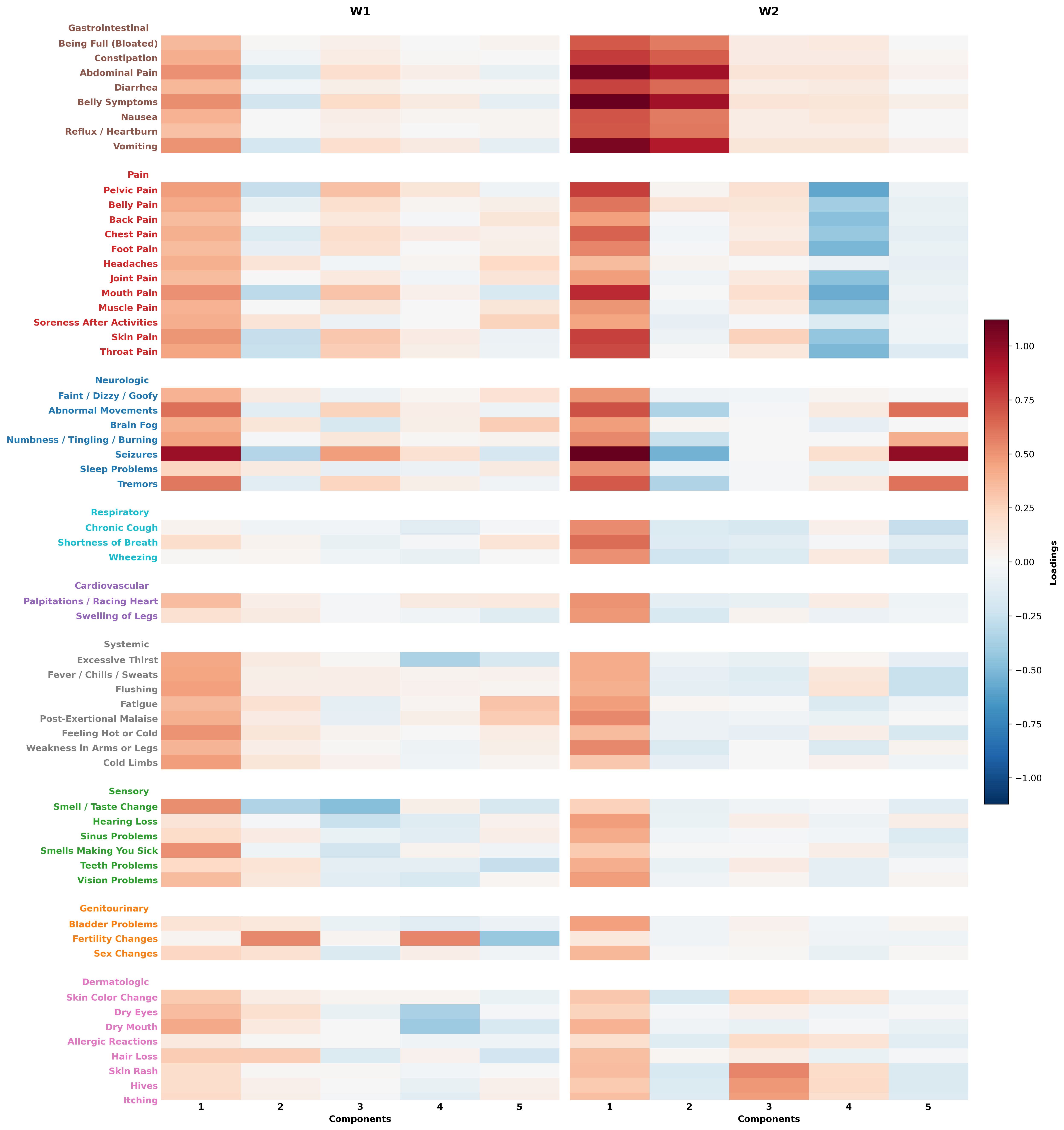}
    \caption{\textbf{Disentangled Static and Dynamic Symptom Phenotypes in the RECOVER Cohort.} Heatmaps of the estimated factor loading matrices $\widehat{\bm{W}}_1$ (between-subject static traits, left) and $\widehat{\bm{W}}_2$ (within-subject dynamic fluctuations, right) from the HPPCA model with $d_1 = d_2 = 5$. The 57 physical symptoms (y-axis) are grouped by physiological system. Color intensity reflects the standardized loading values.}
    \label{fig:recover_heatmaps}
\end{figure}

We grouped the long COVID symptoms into nine categories: gastrointestinal, pain, neurologic, respiratory, cardiovascular, systemic, sensory, genitourinary, and dermatologic. This hierarchical decomposition provides insights into long COVID phenotypes. As shown in Figure \ref{fig:recover_heatmaps}, respiratory symptoms exhibit low absolute loadings in $\bm{W}_1$ across all components but high values in the first component of $\bm{W}_2$. This indicates that respiratory issues are not consistent over time; rather, they are more likely to represent transient acute symptoms or exhibit significant temporal drift. In contrast, symptoms related to pain, gastrointestinal tract, neurology, and genitourinary changes are more likely to exert long-lasting, persistent effects on patients. Furthermore, we observed that the five components of $\bm{W}_2$ represent five distinct symptom profiles that change with time: the first acts as a comprehensive severity score, while the subsequent components isolate gastrointestinal, dermatologic, pain-focused, and neurologic symptom clusters, respectively.

\subsection{Downstream Clinical Prediction}
To validate the clinical utility of the learned latent representations, we evaluated their ability to predict multiple outcomes: the patient-reported outcomes measurement information system (PROMIS) \citep{hays2009development} global health (PROMIS-01) score, an independent 5-category ordinal assessment of overall health; and PROMIS-02 score, another independent 5-category ordinal assessment of quality of life. Similar to previous simulation settings, when HPPCA was constrained to $d_1=d_2=d/2$, the existing methods were granted a matched total latent dimension of $d$.

We adopted a stratified 5-fold cross-validation scheme grouped by subjects. Since mFPCA cannot natively produce visit-specific scores from its formulas, we compared HPPCA against PPCA. For HPPCA, the feature vector was the concatenation of the two-level latent representation $[\widehat{\bm{Z}}_{1,i}^\top, \widehat{\bm{Z}}_{2,ij}^\top]^\top \in \mathbb{R}^{d_1 + d_2}$. We trained four supervised learning classifiers: logistic regression, a linear support vector machine (SVM), random forest, and a histogram-based gradient boosting tree classifier. For the tree-based ensembles, the random forest was configured with 100 estimators, while the gradient boosting model utilized a learning rate of 0.05 and a maximum depth of 4. All models were implemented in Python using \emph{scikit-learn} \citep{scikit-learn}.

Figure~\ref{fig:visit_prediction} shows that HPPCA embeddings achieve higher balanced accuracy than those from PPCA embeddings across logistic regression, a linear SVM and a gradient boosting model for both outcomes. Although PPCA outperforms HPPCA when using the random forest classifier, this specific model yields a lower balanced accuracy overall compared to the other downstream models, suggesting that random forest is not a good fit for this task. We also consider Cohen's kappa as a secondary performance metric, which yields similar conclusions (Web Appendix~D.5 and Web Figure~6). By effectively pooling information across a patient's trajectory, HPPCA gains from aggregated patient information, whereas PPCA treats each visit independently, confounding temporal signals with patient-specific intercepts.

\begin{figure}[!htbp]
    \centering
    \includegraphics[width=\linewidth]{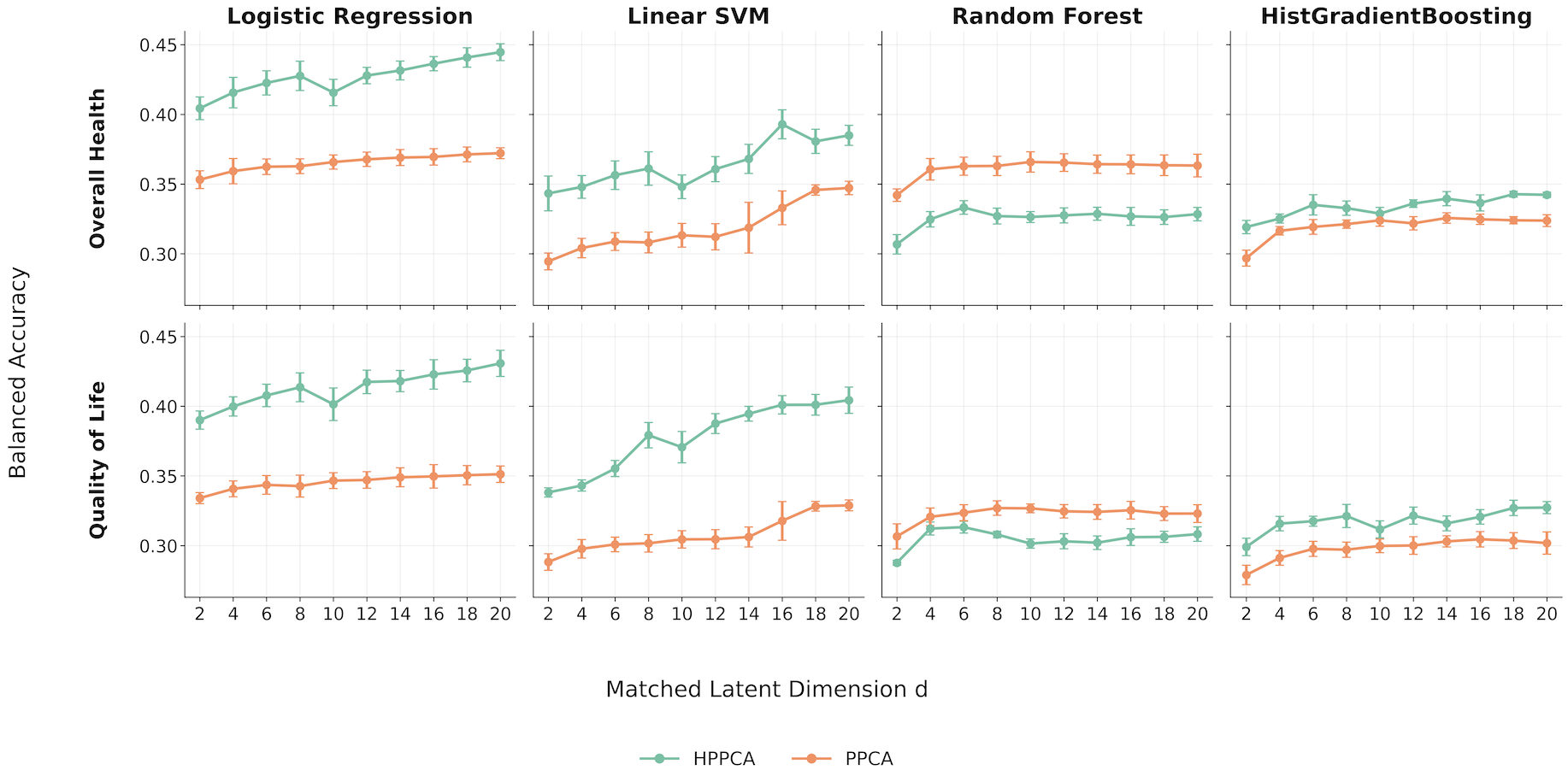} 
    \caption{\textbf{Subject-Visit Level Overall Health and Quality of Life Prediction.} Balanced accuracy for classifying the 5-category PROMIS-01 overall health score and PROMIS-02 quality of life score at specific longitudinal visits. Error bars represent standard deviations across 5-fold cross-validation.}
    \label{fig:visit_prediction}
\end{figure}

\subsection{Missing Data Imputation on Masked Records}
Finally, to quantitatively benchmark HPPCA's ability to recover missing data, we conducted a masking experiment. We randomly masked an additional $20\%$ of the observed entries across the dataset and re-fitted the models. The overall missing rate of this masked dataset is 42.86\%. We benchmarked HPPCA ($d_1=d_2=d/2$) against standard PPCA and the two mFPCA variants, supplying the baselines with a matched total rank of $d$. 

\begin{figure}[!htbp]
    \centering
    \includegraphics[width=0.8\linewidth]{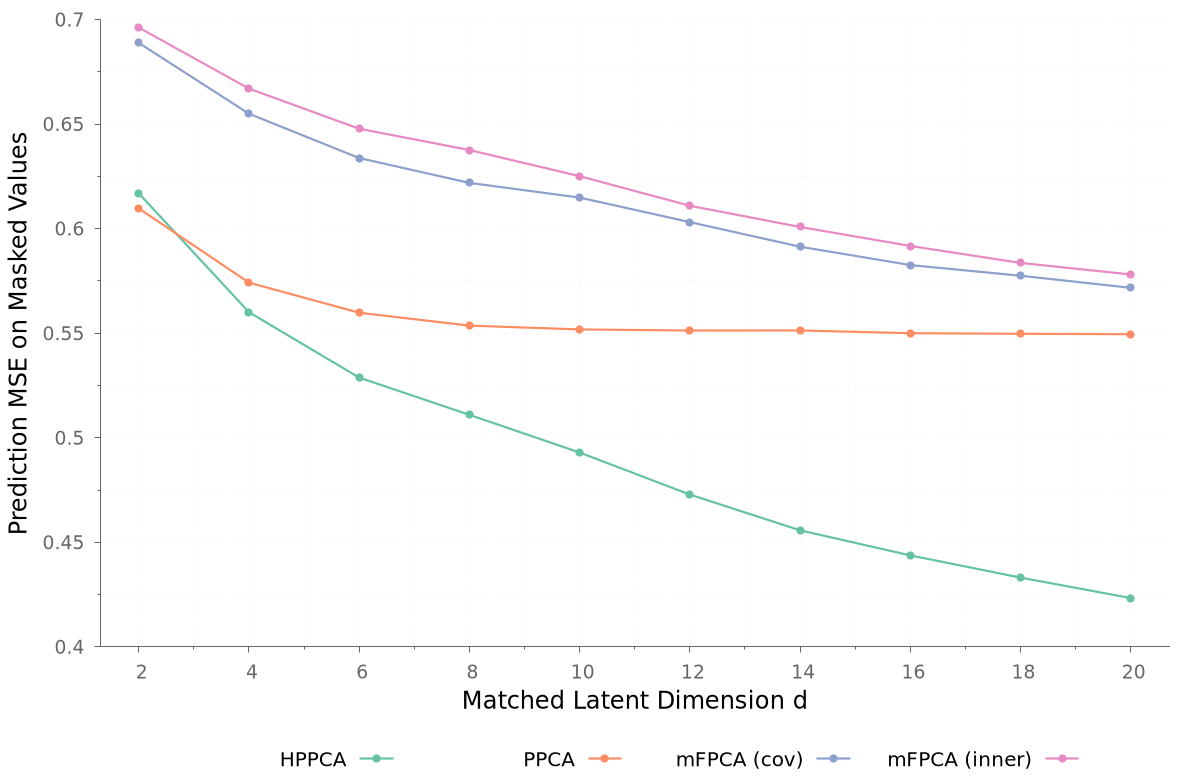} 
    \caption{\textbf{Imputation Accuracy on Masked Clinical Records.} Prediction MSE evaluated exclusively on $20\%$ randomly masked held-out symptoms.}
    \label{fig:recover_imputation}
\end{figure}

As shown in Figure~\ref{fig:recover_imputation}, HPPCA consistently achieves the lowest imputation MSE evaluated strictly on the masked values when the matched dimension $d$ is greater than 2. PPCA gives comparable MSE under $d = 2, 4$; however, when the matched latent dimension continues to increase, the imputation MSE begins to stabilize, indicating that the PPCA model does not capture additional information from the data as effectively as our model does. The two mFPCA variants have a higher imputation MSE than HPPCA and PPCA across all the dimensions, though it continues to drop when the latent dimension increases. 

For this specific task, the mFPCA(inner) method requires over 100 GB of memory for a single model fit. Because its space complexity scales at $O(n^2)$, it is ill-suited for large datasets, particularly compared to the linear $O(n)$ memory cost of HPPCA. Since mFPCA relies on cross-sectional covariance smoothing over time, the high-dimensional irregularity and sparsity of the RECOVER visit schedule cause its functional eigen-decomposition to degrade. In contrast, HPPCA's continuous-time GP structure naturally bridges these observational gaps, utilizing both the subject's static trait $\bm{Z}_{1,i}$ and localized temporal trajectory $\bm{Z}_{2,ij}$ to optimally reconstruct the missing features.

\section{Discussion}\label{sec:discussion}

In this work, we introduced the HPPCA model, as an extension of the PPCA model to longitudinal data. HPPCA is a probabilistic two-level factor model that separates a between-subject structure from time-varying within-subject dynamics. We derived a practical EM procedure that natively accommodates subject-specific irregular observation times, general continuous-time kernels, and heavy missing data, and accelerated it by initializers derived under balanced temporal conditions with mild assumptions on the temporal covariance structure.

As shown in extensive simulation studies, HPPCA accurately recovers the underlying parameter subspaces even under extreme missingness, a sparse regime where traditional functional methods encounter numerical instability. Furthermore, HPPCA delivers substantially lower imputation error than PPCA and multivariate functional PCA (mFPCA) variants given the correct model assumption, particularly when missingness is heavy or temporal correlation is strong. It also demonstrates remarkable robustness to structural misspecification, maintaining good predictive performance even when the true latent dynamics follow discrete-time autoregressive processes. 

In a large longitudinal adult cohort from the NIH RECOVER initiative, HPPCA successfully disentangled chronic baseline vulnerabilities from episodic symptom flare-ups. HPPCA gains increasing model-based variance explained with model dimension. By isolating these distinct clinical phenotypes, the learned joint embeddings significantly improve predictive performance on visit-level downstream tasks relative to PPCA. Finally, HPPCA achieved the lowest imputation mean square error in recovering artificially masked symptom burden records, demonstrating its utility as a interpretable missing data imputation method.

Several future directions naturally emerge from this work. Methodologically, the model can be further generalized by incorporating non-stationary or subject-specific kernels, accommodating heteroskedastic noise, or allowing for non-linear transformations in the loadings. Computationally, although our closed-form initializers significantly accelerate convergence, the EM algorithm can still be intensive for exceptionally large cohorts, so developing faster and more scalable inference procedures is an important avenue.
\newpage


\section*{Acknowledgements}
We thank the participants, investigators, and study staff of the Researching COVID to Enhance Recovery (RECOVER) Initiative. Data for this study were analyzed using the NHLBI BioData Catalyst\textregistered (BDC) ecosystem. We thank the NHLBI for providing access to these datasets and the BDC community for technical support. We are also grateful to Thomas Scott Jr Caldwell for his ideas and effort on accelerating the code. We acknowledge the use of Gemini for grammatical editing and stylistic improvements during the preparation of this manuscript.



\section*{Data Availability}
The data used for this study were accessed through the NHLBI BDC ecosystem (\url{https://biodatacatalyst.nhlbi.nih.gov}). Researchers can request access via the RECOVER Data page by submitting proposals and using NIH eRA Commons credentials.

\section*{Supplementary Material}
Web Appendices, Tables, Figures and Algorithms for Sections 2, 3, 4 are available with this article. We also provide code for implementing the proposed HPPCA method, available with the Supplementary Material and on GitHub repository at \url{https://github.com/dlin-group/HPPCA}.


\newpage

\setcounter{figure}{0}
\setcounter{table}{0}
\setcounter{algorithm}{0}
\setcounter{equation}{0}

\newcounter{webappendix}
\renewcommand{\theequation}{\Alph{webappendix}.\arabic{equation}}
\renewcommand{\thesubsection}{\Alph{webappendix}.\arabic{subsection}}
\renewcommand{\thesubsubsection}{\Alph{webappendix}.\arabic{subsection}.\arabic{subsubsection}}
\captionsetup[figure]{name=Web Figure}
\captionsetup[table]{name=Web Table}
\newcommand{\webappendixsection}[1]{%
  \clearpage
  \refstepcounter{webappendix}%
  \setcounter{equation}{0}%
  \setcounter{subsection}{0}%
  \setcounter{subsubsection}{0}%
  \section*{Web Appendix \Alph{webappendix}: #1}%
}
\floatname{algorithm}{Web Algorithm}
\renewcommand{\thealgorithm}{\arabic{algorithm}}

\webappendixsection{Details of the EM Algorithm for HPPCA}
This supplementary material provides detailed mathematical derivations supporting the EM algorithm presented in Section~2.2, fully accommodating arbitrary item-level missingness, subject-specific observation times $\{t_{ij}\}_{j \in \mathcal{O}_i}$, and unequal number of surveys $J_i$.

\subsection{E-step}
The E-step computes the expectation of the complete-data log-likelihood. The conditional expectation of the complete-data log likelihood given the observed data and current estimated parameters is given by
\begin{align*}
    & \mathbb{E}\big[\, l_{i}(\bm{\theta}) \mid \bm{Y}_i^{(o)}, \bm{\theta}^{(t)}\,\big]  \\
     = &  - \frac{pJ_i}{2} \log(2\pi \sigma^2) - \frac{1}{2} \tr(\langle \bm{Z}_{1, i} \bm{Z}_{1, i}^\top \rangle) - \frac{1}{2}\sum_{r = 1}^{d_2}\log (\det(\bm{\Sigma}_i^{(r)})) -  \frac{1}{2}\sum_{r = 1}^{d_2} \tr(\langle \bm{z}_{2, i}^{(r)}  \bm{z}_{2, i}^{(r)\top}\rangle (\bm{\Sigma}^{(r)}_i)^{-1})\\
& - \sum_{j \in \mathcal{O}_i} \bigg \{\frac{1}{2\sigma^2}\bm{Y}_{ij}^{(o)\top} \bm{Y}_{ij}^{(o)} + \frac{1}{2\sigma^2}\tr(\langle\bm{Y}_{ij}^{(m)} \bm{Y}_{ij}^{(m)\top}\rangle) + 
\frac{1}{2\sigma^2}\tr(\bm{W}_1^\top \bm{W}_1 \langle \bm{Z}_{1,i} \bm{Z}_{1, i}^\top\rangle)  \\
& + 
\frac{1}{2\sigma^2}\tr(\bm{W}_2^\top \bm{W}_2 \langle \bm{Z}_{2,ij} \bm{Z}_{2, ij}^\top\rangle) + \frac{1}{\sigma^2}\tr(\bm{W}_1^\top \bm{W}_2 \langle \bm{Z}_{2,ij} \bm{Z}_{1, i}^\top\rangle) - \frac{1}{\sigma^2} \langle \bm{Z}_{1, i}^\top  \rangle \bm{W}_1^{(o, ij)\top}\bm{Y}_{ij}^{(o)}\\
& - \frac{1}{\sigma^2} \langle \bm{Z}_{2, ij}^\top  \rangle \bm{W}_2^{(o, ij)\top}\bm{Y}_{ij}^{(o)}- \frac{1}{\sigma^2} \tr(\bm{W}_1^{(m, ij)\top}\langle \bm{Y}_{ij}^{(m)} \bm{Z}_{1, i}^\top\rangle ) - \frac{1}{\sigma^2} \tr(\bm{W}_2^{(m, ij)\top}\langle \bm{Y}_{ij}^{(m)} \bm{Z}_{2, ij}^\top \rangle )
\bigg \}.
\end{align*}

With the posterior distribution of the latent factors established in Section~2.2, we calculate the conditional moments involving the missing data. By taking the expectation of the generative model $\bm{Y}_{ij}^{(m)} = \bm{W}_1^{(m,ij)}\bm{Z}_{1,i} + \bm{W}_2^{(m,ij)}\bm{Z}_{2,ij} + \sigma \bm{\varepsilon}_{ij}^{(m)}$ conditioned on the observed data $\bm{Y}_i^{(o)}$, we obtain
\begin{align*}
    \langle \bm{Y}_{ij}^{(m)}\rangle  &= \bm{W}_1^{(m, ij)} \langle \bm{Z}_{1, i} \rangle + \bm{W}_2^{(m, ij)}\langle\bm{Z}_{2, ij}\rangle, \\
   \langle \bm{Y}_{ij}^{(m)}\bm{Y}_{ij}^{(m)\top}\rangle  &= \bm{W}_1^{(m, ij)} \langle\bm{Z}_{1, i} \bm{Z}_{1, i}^\top \rangle \bm{W}_1^{(m, ij)\top} + \bm{W}_2^{(m, ij)} \langle\bm{Z}_{2, ij} \bm{Z}_{2,ij}^\top \rangle \bm{W}_2^{(m, ij)\top}\\
    &\quad + \bm{W}_1^{(m, ij)} \langle\bm{Z}_{1, i} \bm{Z}_{2,ij}^\top \rangle \bm{W}_2^{(m, ij)\top} + \big(\bm{W}_1^{(m, ij)} \langle\bm{Z}_{1, i} \bm{Z}_{2,ij}^\top \rangle \bm{W}_2^{(m, ij)\top}\big)^\top + \sigma^2 \bm{I}_{p_{ij}^{(m)}},\\
    \langle \bm{Y}_{ij}^{(m)} \bm{Z}_{1, i}^\top\rangle  &= \bm{W}_1^{(m, ij)} \langle \bm{Z}_{1, i}\bm{Z}_{1, i}^\top \rangle + \bm{W}_2^{(m, ij)} \langle \bm{Z}_{2, ij}\bm{Z}_{1, i}^\top \rangle, \\
    \langle \bm{Y}_{ij}^{(m)} \bm{Z}_{2, ij}^\top\rangle &= \bm{W}_1^{(m, ij)} \langle \bm{Z}_{1, i}\bm{Z}_{2, ij}^\top \rangle + \bm{W}_2^{(m, ij)} \langle \bm{Z}_{2, ij}\bm{Z}_{2, ij}^\top \rangle.
\end{align*}
where $p_{ij}^{(m)}$ is the number of missing entries in $\bm{Y}_{ij}$.
\subsection{M-step}
Let $\bm{m}_{ij} \in \{0,1\}^p$ be the missingness mask for observation $\bm{Y}_{ij}$, where $m_{ij,k}=1$ if the $k$th entry is missing. To evaluate the M-step updates seamlessly, we construct the full pseudo-data cross-moment matrices $\widetilde{\bm{G}}_{ij}^{(1)} = \langle \bm{Y}_{ij}\bm{Z}_{1, i}^\top \rangle \in \mathbb{R}^{p \times d_1}$ and $\widetilde{\bm{G}}_{ij}^{(2)} = \langle \bm{Y}_{ij} \bm{Z}_{2, ij}^\top \rangle \in \mathbb{R}^{p \times d_2}$ for the $j$th survey of the $i$th subject. The elements of these matrices are assembled by explicitly mapping the partitioned observed and missing components back to their original feature dimensions based on the missingness mask $\bm{m}_{ij}$.

For the $k$th feature ($k = 1, \dots, p$) and $r$th subject level latent dimension ($r = 1, \dots, d_1$), the $(k, r)$th element of $\widetilde{\bm{G}}_{ij}^{(1)}$ is defined as:
\[
\big(\widetilde{\bm{G}}^{(1)}_{ij}\big)_{kr} = 
\begin{cases}
    Y_{ij, k'}^{(o)} \langle \bm{Z}_{1, i} \rangle_r &  \text{if } m_{ij, k} = 0, \\[2mm]
    \big(\langle \bm{Y}_{ij}^{(m)} \bm{Z}_{1, i}^\top \rangle\big)_{k''r} & \text{if } m_{ij, k} = 1,
\end{cases}
\]
where $k'$ is the sequential index in the observed sub-vector $\bm{Y}_{ij}^{(o)}$ that corresponds to the original $k$th feature of $\bm{Y}_{ij}$, and $k''$ is the sequential index in the missing sub-vector $\bm{Y}_{ij}^{(m)}$ that corresponds to the original $k$th feature.

Likewise, for the $k$th feature ($k = 1, \dots, p$) and $r$th subject-visit level latent dimension ($r = 1, \dots, d_2$), the $(k, r)$th element of the cross-moment matrix $\widetilde{\bm{G}}_{ij}^{(2)}$ is constructed as:
\[
\big(\widetilde{\bm{G}}^{(2)}_{ij}\big)_{kr} = 
\begin{cases}
    Y_{ij, k'}^{(o)} \langle \bm{Z}_{2, ij} \rangle_r &  \text{if } m_{ij, k} = 0, \\[2mm]
    \big(\langle \bm{Y}_{ij}^{(m)} \bm{Z}_{2, ij}^\top \rangle\big)_{k''r} & \text{if } m_{ij, k} = 1.
\end{cases}
\]

 In the M-step, we maximize $Q(\bm{\theta} \mid \bm{\theta}^{(t)})$. We take the derivative of the $Q$-function with respect to $\bm{W}_1$ and $\bm{W}_2$ and set them to zero:
\begin{align*}
    \frac{\partial Q}{\partial \bm{W}_1} &= \frac{1}{\sigma^2}\sum_{i=1}^n \sum_{j \in \mathcal{O}_i} \Big( \widetilde{\bm{G}}_{ij}^{(1)} - \bm{W}_1 \langle \bm{Z}_{1, i} \bm{Z}_{1, i}^\top \rangle - \bm{W}_2 \langle \bm{Z}_{2, ij} \bm{Z}_{1, i}^\top \rangle \Big) = \bm{0}, \\
    \frac{\partial Q}{\partial \bm{W}_2} &= \frac{1}{\sigma^2}\sum_{i=1}^n \sum_{j \in \mathcal{O}_i} \Big( \widetilde{\bm{G}}_{ij}^{(2)} - \bm{W}_1 \langle \bm{Z}_{1, i} \bm{Z}_{2, ij}^\top \rangle - \bm{W}_2 \langle \bm{Z}_{2, ij} \bm{Z}_{2, ij}^\top \rangle \Big) = \bm{0}.
\end{align*}
Factoring $[\bm{W}_1, \bm{W}_2]$ out of the sums natively recovers the block matrix system presented in equation~(8) of the main text. Fixing the updated $\widetilde{\bm{W}}_1$ and $\widetilde{\bm{W}}_2$, we take the derivative of the $Q$-function with respect to $\sigma^2$ and setting to zero: 
\begin{equation*}
     \frac{\partial Q}{\partial \sigma^2} = -\frac{1}{\sigma^2}\sum_{i = 1}^n \frac{pJ_i}{2}  + \frac{1}{2\sigma^4}\sum_{i = 1}^n \sum_{j \in \mathcal{O}_i}\langle \|\bm{Y}_{ij} - \widetilde{\bm{W}}_1 \bm{Z}_{1, i} - \widetilde{\bm{W}}_2 \bm{Z}_{2, ij}\|^2_2 \rangle  = 0.
\end{equation*}
This gives the closed form update of $\sigma^2$ shown in equation.~(9) of the main text.

\subsection{M-step: Kernel Length-Scale Updates}
\label{webapp:sec:kernel_updates}

Our framework gracefully accommodates both a shared temporal length-scale across all dynamic dimensions and dimension-specific length-scales. Because the temporal priors are structurally isomorphic, we unify their M-step Newton-Raphson updates to avoid mathematical redundancy. Let $\ell$ denote the target generic length-scale parameter being updated. We define its corresponding effective dimensionality $d^*$ and the pooled posterior second-moment matrix $\bm{C}_{2,i}^* \in \mathbb{R}^{J_i \times J_i}$ as follows:
\begin{itemize}
    \item \textbf{Dimension-Specific Length-Scales ($\ell^{(r)}$):} We update each dimension $r \in \{1, \dots, d_2\}$ independently. We set $d^* = 1$ and $\bm{C}_{2,i}^* = \langle \bm{z}_{2, i}^{(r)}  \bm{z}_{2, i}^{(r)\top} \rangle$.
    \item \textbf{Shared Length-Scale ($\ell$):} We pool information across all $d_2$ dimensions. We set $d^* = d_2$ and $\bm{C}_{2,i}^* = \sum_{r=1}^{d_2} \langle \bm{z}_{2, i}^{(r)}  \bm{z}_{2, i}^{(r)\top} \rangle$.
\end{itemize}

Given the specific formulation of $\ell$, $d^*$, and $\bm{C}_{2,i}^*$, the relevant portion of the $Q$-function simplifies to a single objective:
\begin{align*}
    Q(\ell) = - \frac{1}{2}\sum_{i=1}^n \left\{ d^* \log \det(\bm{\Sigma}_i) +  \operatorname{tr}\!\big(\bm{C}_{2,i}^* \bm{\Sigma}_i^{-1} \big) \right\}.
\end{align*}
To streamline the matrix calculus, we define the auxiliary residual precision matrix $\bm{A}_{i} = \bm{\Sigma}_i^{-1}\bigl(\bm{C}_{2,i}^* - d^*\bm{\Sigma}_i\bigr)\bm{\Sigma}_i^{-1}$. The exact analytical gradient equals to
\begin{align*}
    g(\ell) &= \frac{\partial Q(\ell)}{\partial \ell} = \frac{1}{2} \sum_{i=1}^n \operatorname{tr}\!\Big( \bm{A}_{i} \frac{\partial \bm{\Sigma}_i}{\partial \ell} \Big).
\end{align*}

\paragraph{Radial Basis Function (RBF) Kernel.}
For the RBF kernel, the exact first derivative of the temporal covariance matrix is $\partial_\ell \bm{\Sigma}_i = \ell^{-3}(\bm{\Sigma}_i \odot \bm{\Delta}_i)$, where $(\bm{\Delta}_i)_{a,b} = (t_{i, j_a} - t_{i, j_b})^2$, and $\odot$ denotes the Hadamard product. The gradient is elegantly expressed as $g(\ell) = \frac{1}{2\ell^3} \sum_i \operatorname{tr}(\bm{A}_i\bm{B}_i)$, where $\bm{B}_{i} = \bm{\Sigma}_i \odot \bm{\Delta}_i$. 

By applying the cyclic permutation property of the trace operator, the analytical Hessian $h(\ell) = \partial g(\ell)/\partial \ell$ precisely simplifies to
\begin{align*}
    h(\ell) =& - \frac{3}{\ell} g(\ell) + \frac{1}{2\ell^6}\sum_{i=1}^n\Big[ \operatorname{tr}\big\{ \bm{A}_{i}(\bm{\Sigma}_i \odot \bm{\Delta}_i^{\odot 2})\big\} \\
    & \quad - 2\operatorname{tr}\big( \bm{\Sigma}_i^{-1}\bm{B}_{i} \bm{A}_{i} \bm{B}_{i}\big) - d^* \operatorname{tr}\big( \bm{\Sigma}_i^{-1}\bm{B}_{i} \bm{\Sigma}_i^{-1} \bm{B}_{i}\big) \Big],
\end{align*}
where $\bm{\Delta}_i^{\odot 2}$ denotes the element-wise square of the temporal distance matrix.

\paragraph{Mat\'ern $\nu=5/2$ Kernel.}
Let $d_{ab} = |t_{i, j_a} - t_{i, j_b}|$ denote the Euclidean temporal distance. The exact first derivative matrix $\bm{K}_i^{(1)} = \partial_\ell \bm{\Sigma}_i$ equals to
\begin{equation*}
    (\bm{K}_i^{(1)})_{a,b} = \exp\!\left(-\frac{\sqrt{5}d_{ab}}{\ell}\right) \frac{5d_{ab}^2}{3\ell^4} (\ell + \sqrt{5}d_{ab}).
\end{equation*}
The exact second derivative matrix $\bm{K}_i^{(2)} = \partial^2_\ell \bm{\Sigma}_i$ expands to
\begin{equation*}
    (\bm{K}_i^{(2)})_{a,b} = \exp\!\left(-\frac{\sqrt{5}d_{ab}}{\ell}\right) \frac{5d_{ab}^2}{3\ell^6} \Big( 5d_{ab}^2 - 3\sqrt{5}d_{ab}\ell - 3\ell^2 \Big).
\end{equation*}
Using the gradient $g(\ell) = \frac{1}{2}\sum_i \operatorname{tr}(\bm{A}_{i} \bm{K}_i^{(1)})$ and leveraging identical trace cyclic simplifications, the exact Hessian is written as:
\begin{align*}
    h(\ell) &= \frac{1}{2}\sum_{i=1}^{n} \Big\{ \operatorname{tr}\big(\bm{A}_{i} \bm{K}_i^{(2)}\big) - 2\operatorname{tr}\big(\bm{\Sigma}_i^{-1}\bm{K}_i^{(1)}\bm{A}_{i}\bm{K}_i^{(1)}\big) - d^* \operatorname{tr}\big(\bm{\Sigma}_i^{-1}\bm{K}_i^{(1)}\bm{\Sigma}_i^{-1}\bm{K}_i^{(1)}\big) \Big\}.
\end{align*}

By substituting the appropriate $d^*$ and $\bm{C}_{2,i}^*$, they provide Newton-Raphson updates $\widetilde{\ell} = \ell - g(\ell)/h(\ell)$ for arbitrary parameterizations and varying subject-specific grids $J_i$.

\webappendixsection{Proofs and Algorithms}
\label{supp:proofs}

\subsection{Proof of Lemma 1}

\begin{proof}[Proof of Lemma~1]
We first use the standard property of the vectorization operator, $\vecop(\bm A\bm B\bm C) = (\bm C^\top\otimes \bm A)\vecop(\bm B)$. Applying this to the rotation $\widetilde{\bm Y}_i = \bm I_p \bm Y_i \bm U_T$, we have
\[
\vecop(\widetilde{\bm Y}_i)=\vecop(\bm I_p \bm Y_i\bm U_T)
=
(\bm U_T^\top\otimes \bm I_p)\bm y_i.
\]
Because $\widetilde{\bm Y}_i$ is a linear transformation of the Gaussian random vector $\bm y_i$, it remains jointly Gaussian. Under the surrogate model, the covariance of $\bm y_i$ is 
\[
\Cov(\bm y_i) = \one_J\one_J^\top \otimes \A + \bm{\Sigma}\otimes \B + \sigma^2 \bm I_{pJ}.
\]
Therefore, the covariance of the rotated data is
\begin{align*}
\Cov\!\big[\vecop(\widetilde{\bm Y}_i)\big]
&=
(\bm U_T^\top\otimes \bm I_p)
\Cov(\bm y_i)
(\bm U_T\otimes \bm I_p).
\end{align*}
Applying the mixed-product property of Kronecker products, $(\bm A \otimes \bm B)(\bm C \otimes \bm D) = (\bm A\bm C) \otimes (\bm B\bm D)$, the covariance matrix distributes as follows:
\begin{align*}
\Cov\!\big[\vecop(\widetilde{\bm Y}_i)\big]
&=
(\bm U_T^\top\one_J\one_J^\top\bm U_T)\otimes (\bm I_p \A \bm I_p)
+
(\bm U_T^\top\bm\Sigma\bm U_T)\otimes (\bm I_p \B \bm I_p)
+
\sigma^2 (\bm U_T^\top \bm I_J \bm U_T \otimes \bm I_p \bm I_p)\\
&=
(\bm U_T^\top\one_J\one_J^\top\bm U_T)\otimes \A
+
(\bm U_T^\top\bm\Sigma\bm U_T)\otimes \B
+
\sigma^2 \bm I_{pJ}.
\end{align*}
By Assumption~1, $\bm U_T$ is an orthogonal matrix whose first column is $\bm u_1 = \one_J/\sqrt{J}$, and its columns are the eigenvectors of $\bm \Sigma$. Therefore, the matrix-vector product $\bm U_T^\top \one_J$ becomes a vector with $\sqrt{J}$ in the first coordinate and $0$ elsewhere. Its outer product with itself gives a diagonal matrix with exactly one non-zero entry:
\[
\bm U_T^\top\one_J\one_J^\top\bm U_T=\diag(J,0,\ldots,0).
\]
since $\bm U_T$ diagonalizes $\bm \Sigma$ and $\bm U_T^\top \bm U_T = \bm I_J$, we have
\[
\bm U_T^\top\bm\Sigma\bm U_T=\diag(\lambda_1,\ldots,\lambda_J).
\]
Substituting these diagonal matrices into our expression for the covariance yields a perfectly block-diagonal structure:
\[
\Cov\!\big[\vecop(\widetilde{\bm Y}_i)\big]
=
\diag\!\Big(
J\A+\lambda_1\B+\sigma^2 \I_p,\,
\lambda_2\B+\sigma^2 \I_p,\,
\ldots,\,
\lambda_J\B+\sigma^2 \I_p
\Big).
\]
Because $\vecop(\widetilde{\bm Y}_i)$ is jointly Gaussian and its covariance matrix is block-diagonal, the rotated columns $\{\widetilde{\bm Y}_{it}\}_{t=1}^J$ are mutually independent. Their marginal distributions directly correspond to the diagonal blocks.

Due to this mutual independence, the joint Gaussian negative log-likelihood factors cleanly into a sum of independent negative log-likelihoods for each column:
\[
\mathcal L(\A,\B,\sigma^2)
=
\mathcal L_A(\A\mid \B,\sigma^2)
+
\mathcal L_B(\B,\sigma^2)
+
\text{constant},
\]
with the mean-direction $\mathcal L_A$ corresponding to $t=1$ and the contrast-direction $\mathcal L_B$ grouping the components for $t=2,\ldots,J$, as given by Section~2.
\end{proof}

\subsection{Proof of Theorem 1}

\begin{proof}[Proof of Theorem~1]
For fixed $\B$ and $\sigma^2>0$, we define the background covariance for the mean direction as
\[
\bm H=\lambda_1\B+\sigma^2 \I_p.
\]
Because $\sigma^2 > 0$ and $\B \succeq 0$, $\bm H$ is strictly positive definite ($\bm H \succ 0$). To simplify the optimization over $\A$, we use a transformation to standardize the parameter space. Let
\[
\bm T=\bm H^{-1/2}\bm S_1\bm H^{-1/2},
\qquad
\bm M=J\bm H^{-1/2}\A\bm H^{-1/2}.
\]
Since $\bm H^{-1/2}$ is symmetric and strictly positive definite, the original constraints $\A \succeq 0$ and $\rank(\A) \le d_1$ translate directly to $\bm M\succeq 0$ and $\rank(\bm M)\le d_1$. We can rewrite the marginal covariance for $t=1$ as
\[
J\A+\bm H=\bm H^{1/2}\bm M\bm H^{1/2} + \bm H^{1/2} \I_p \bm H^{1/2} = \bm H^{1/2}(\bm M + \I_p)\bm H^{1/2}.
\]
We now express the objective $\mathcal L_A$ in terms of $\bm M$. Using the multiplicative property of determinants, we have
\[
\log\det(J\A+\bm H) = \log\det(\bm H) + \log\det(\bm M + \I_p).
\]
For the trace term, we use the cyclic permutation property of the trace:
\begin{align*}
\tr\!\big[\bm S_1(J\A+\bm H)^{-1}\big] &= \tr\!\big[\bm S_1 \bm H^{-1/2}(\bm M + \I_p)^{-1}\bm H^{-1/2}\big] \\
&= \tr\!\big[\bm H^{-1/2}\bm S_1\bm H^{-1/2}(\bm M + \I_p)^{-1}\big] \\
&= \tr\!\big[\bm T(\bm M + \I_p)^{-1}\big].
\end{align*}
Thus, up to the additive constant $(n/2)\log\det(\bm H)$, which is independent of $\A$, minimizing $\mathcal L_A$ is equivalent to minimizing the surrogate objective
\[
\psi(\bm M)
=
\log\det(\bm M + I_p)
+
\tr\!\big[\bm T(\bm M + I_p)^{-1}\big]
\]
over $\bm M\succeq 0$ with $\rank(\bm M)\le d_1$.

To explicitly enforce the positive semi-definiteness and rank constraints, we parameterize $\bm M$ using a factor loading matrix $\bm W \in \mathbb{R}^{p \times d_1}$ such that $\bm M = \bm W\bm W^\top$. The optimization problem then becomes minimizing
\[
\psi(\bm W) = \log\det(\bm W\bm W^\top + \I_p) + \tr\!\big[\bm T(\bm W\bm W^\top + \I_p)^{-1}\big]
\]
over the unconstrained matrix $\bm W \in \mathbb{R}^{p \times d_1}$.

Crucially, this objective function is mathematically identical to the negative log-likelihood of a PPCA model \citep{tipping1999probabilistic}, where $\bm T$ plays the role of the empirical sample covariance matrix, $d_1$ is the latent dimensionality, and the isotropic noise variance is explicitly fixed at $\sigma^2 = 1$. 

Taking the derivative of $\psi(\bm W)$ with respect to $\bm W$ and equating it to zero yields the exact stationary condition:
\[
\bm T(\bm W\bm W^\top + \I_p)^{-1}\bm W = \bm W.
\]
Let the eigendecomposition of the whitened sample covariance be
\[
\bm T=\bm V\diag(t_1,\ldots,t_p)\bm V^\top,
\qquad
t_1\ge\cdots\ge t_p\ge 0.
\]
As proven by \citet{tipping1999probabilistic}, the stable minimum of the objective is achieved when the columns of $\bm W$ span the principal eigenspace of the sample covariance $\bm T$. The optimal loading matrix is given by:
\[
\bm W^\star = \bm V_{d_1} (\bm \Lambda_{d_1} - \I_{d_1})_+^{1/2} \bm R,
\]
where $\bm V_{d_1} \in \mathbb{R}^{p \times d_1}$ contains the eigenvectors corresponding to the largest $d_1$ eigenvalues of $\bm T$, $\bm \Lambda_{d_1} = \diag(t_1, \ldots, t_{d_1})$, and $\bm R$ is an arbitrary $d_1 \times d_1$ orthogonal rotation matrix. The positive thresholding operation $(\cdot)_+ = \max(\cdot, 0)$ guarantees that the singular values remain real-valued, which corresponds to setting the latent variance to zero for any dimensions where the sample variance $t_k$ does not exceed the assumed noise variance of $1$.

Returning to our original surrogate matrix $\bm M^\star = \bm W^\star \bm W^{\star\top}$, the arbitrary rotation $\bm R$ exactly cancels out ($\bm R \bm R^\top = \I_{d_1}$), yielding a unique global solution:
\[
\bm M^\star = \bm V_{d_1} (\bm \Lambda_{d_1} - \I_{d_1})_+ \bm V_{d_1}^\top.
\]
This expression precisely corresponds to applying the positive spectral thresholding operator to $\bm T - I_p$, which retains the largest $d_1$ positive eigenvalues and sets the rest to zero:
\[
\bm M^\star = \Pi_+^{(d_1)}(\bm T - \I_p).
\]
Transforming back to the original parameterization of $\A$ yields:
\[
\A^\star
=
\frac1J\bm H^{1/2}\bm M^\star\bm H^{1/2}
=
\frac1J\bm H^{1/2}
\Pi_+^{(d_1)}
\!\left(
\bm H^{-1/2}\bm S_1\bm H^{-1/2}-\I_p
\right)
\bm H^{1/2},
\]
which proves the result.
\end{proof}

\subsection{Contrast block under Assumption 1}

Write the spectral decomposition of $\B$ as
\[
\B=\Q\diag(b_1,\ldots,b_p)\Q^\top,
\qquad
\Q^\top\Q=\I_p,
\qquad
b_k\ge 0.
\]
For fixed $\Q$, define the rotated variance components
\[
\D_t=\Q^\top \bm S_t \Q,
\qquad
s_{t,k}=(\D_t)_{kk},
\qquad
a_{t,k}=\lambda_t b_k+\sigma^2,
\qquad t=2,\ldots,J.
\]
By factoring out $\Q$, the log-determinant and trace operators in the contrast negative log-likelihood act only on diagonal matrices. Thus, up to a positive constant factor, the contrast contribution is given by
\[
\Phi(\bb,\sigma^2;\Q)
=
\sum_{k=1}^p\sum_{t=2}^J
\left(
\log a_{t,k}+\frac{s_{t,k}}{a_{t,k}}
\right).
\]

\begin{proposition}[Score equations for $(\bb,\sigma^2)$ at fixed $\Q$]
\label{prop:score-bsigma}
For fixed orthogonal $\Q$, the partial derivatives of $\Phi(\bb,\sigma^2;\Q)$ are
\[
\frac{\partial \Phi}{\partial b_k}
=
\sum_{t=2}^J
\lambda_t
\left(
\frac{1}{a_{t,k}}-\frac{s_{t,k}}{a_{t,k}^2}
\right),
\qquad
k=1,\ldots,p,
\]
and
\[
\frac{\partial \Phi}{\partial \sigma^2}
=
\sum_{k=1}^p\sum_{t=2}^J
\left(
\frac{1}{a_{t,k}}-\frac{s_{t,k}}{a_{t,k}^2}
\right).
\]
Hence, any interior stationary point satisfies
\[
\sum_{t=2}^J \lambda_t w_{t,k}\big(s_{t,k}-\lambda_t b_k-\sigma^2\big)=0,
\qquad k=1,\ldots,p,
\]
and
\[
\sum_{k=1}^p\sum_{t=2}^J
w_{t,k}\big(s_{t,k}-\lambda_t b_k-\sigma^2\big)=0,
\qquad
w_{t,k}=a_{t,k}^{-2}.
\]
These are exactly the weighted normal equations associated with the pseudo-regression
\[
s_{t,k}\approx \lambda_t b_k+\sigma^2,
\]
with one common intercept $\sigma^2$ shared across $k$.
\end{proposition}

\begin{proof}[Proof of Proposition~1]
By standard scalar calculus, the derivative of the core term with respect to $a$ is
\[
\frac{d}{da}\left(\log a+\frac{s}{a}\right)=\frac{1}{a}-\frac{s}{a^2}.
\]
Applying the chain rule with the partial derivatives of the linear formulation $a_{t,k} = \lambda_t b_k + \sigma^2$:
\[
\frac{\partial a_{t,k}}{\partial b_k}=\lambda_t,
\qquad
\frac{\partial a_{t,k}}{\partial \sigma^2}=1,
\]
directly yields the stated gradients. By factoring out $1/a_{t,k}^2$, we algebraically rearrange the terms inside the summation:
\[
\frac{1}{a_{t,k}}-\frac{s_{t,k}}{a_{t,k}^2}
=
\frac{a_{t,k}-s_{t,k}}{a_{t,k}^2}
=
-a_{t,k}^{-2}\big(s_{t,k}-a_{t,k}\big).
\]
Substituting $w_{t,k} = a_{t,k}^{-2}$ and $a_{t,k} = \lambda_t b_k + \sigma^2$ yields $-w_{t,k}\big(s_{t,k}-\lambda_t b_k-\sigma^2\big)$.
Setting the partial derivatives to zero and substituting this rearranged form produces exactly the weighted least squares normal equations, matching an iteratively reweighted least squares (IRLS) update step.
\end{proof}

For fixed $(\bb,\sigma^2)$, let
\[
\V_t=\diag\!\left(\frac{1}{a_{t,1}},\ldots,\frac{1}{a_{t,p}}\right),
\qquad t=2,\ldots,J,
\]
and define the objective function restricted to the Stiefel manifold of orthogonal matrices \citep{edelman1998geometry,absil2008optimization}:
\[
f(\Q)=\sum_{t=2}^J \tr\!\big(\Q^\top \bm S_t \Q \V_t\big),
\qquad \Q^\top\Q=\I_p.
\]

\begin{proposition}[Stiefel gradient]
\label{prop:stiefel}
The Euclidean gradient of $f$ is
\[
\nabla f(\Q)=2\sum_{t=2}^J \bm S_t \Q \V_t.
\]
Under the embedded Euclidean metric on the Stiefel manifold, the Riemannian gradient is
\[
\grad f(\Q)
=
\nabla f(\Q)-\Q\,\text{Sym}[\Q^\top \nabla f(\Q)].
\]
Consequently, a first-order stationary point satisfies
\[
\sum_{t=2}^J \bm S_t \Q \V_t = \Q\bm \Lambda
\]
for some symmetric matrix $\bm \Lambda=\bm \Lambda^\top$.
\end{proposition}

\begin{proof}[Proof of Proposition~2]
We compute the Euclidean gradient by considering the matrix differential of the trace. Using the product rule for a single component $t$:
\[
df_t(\Q) = \tr\!\big(d\Q^\top \bm S_t \Q \V_t + \Q^\top \bm S_t d\Q \V_t\big).
\]
Because both $\bm S_t$ and $\V_t$ are symmetric matrices and the trace is invariant under transpose and cyclic permutations, the two differential terms are identical:
\[
\tr\!\big(d\Q^\top \bm S_t \Q \V_t\big) = \tr\!\big\{(\bm S_t \Q \V_t)^\top d\Q\big\},
\]
and
\[
\tr\!\big(\Q^\top \bm S_t d\Q \V_t\big) = \tr\!\big(\V_t \Q^\top \bm S_t d\Q\big) = \tr\!\big[(\bm S_t \Q \V_t)^\top d\Q\big].
\]
Thus, the differential is $df_t(\Q) = \tr\!\big[(2 \bm S_t \Q \V_t)^\top d\Q \big]$. Summing over $t$ and matching with the definition of the Euclidean gradient $df = \tr\big[\nabla f(\Q)^\top d\Q\big]$ yields $\nabla f(\Q)=2\sum_{t=2}^J \bm S_t \Q \V_t$.

The Stiefel manifold is the set of orthogonal matrices $\mathcal{M} = \{\Q \in \mathbb{R}^{p \times p} \mid \Q^\top \Q = I_p\}$. The Riemannian gradient is obtained by orthogonally projecting the Euclidean gradient onto the tangent space of this manifold. Under the embedded Euclidean metric, this projection is given by subtracting the normal component:
\[
\grad f(\Q)
=
\nabla f(\Q)-\Q\,\text{Sym}[\Q^\top \nabla f(\Q)],
\]
where $\text{Sym}(\bm X) = \frac{1}{2}(\bm X + \bm X^\top)$. Setting $\grad f(\Q)=0$ indicates a stationary point on the manifold, giving
\[
\nabla f(\Q)=\Q\,\text{Sym}[\Q^\top \nabla f(\Q)].
\]
Defining the explicitly symmetric matrix $\bm \Lambda = \frac{1}{2}\text{Sym}[\Q^\top \nabla f(\Q)]$, this is equivalent to $\nabla f(\Q) = \Q(2\bm \Lambda)$, which naturally reduces to the stated stationarity condition $\sum_{t=2}^J \bm S_t \Q \V_t = \Q\bm \Lambda$.
\end{proof}

\subsection{Web Algorithm 1: Initializer 1 under known shared temporal covariance and Assumption 1}
\label{supp:algorithms}

\begin{algorithm}[H]
\caption{Initializer 1: block-coordinate surrogate MLE under Assumption~1}
\label{alg:init1}
\begin{algorithmic}[1]
\State \textbf{Input:} rotated covariances $\{\bm S_t\}_{t=1}^J$, temporal eigenvalues $\{\lambda_t\}_{t=1}^J$, ranks $(d_1,d_2)$.
\State \textbf{Initialize:} $\Q^{(0)}\leftarrow$ leading eigenvectors of $\overline{\bm S}_c$; $\bb^{(0)}\ge 0$ while explicitly restricting $b_k^{(0)} = 0$ for $k > d_2$; $\sigma^{2(0)}>\varepsilon$.
\For{$m=0,1,2,\ldots$ until convergence}
    \State \textbf{$(\bb,\sigma^2)$-step:} with $\Q=\Q^{(m)}$, update active components $(b_1, \ldots, b_{d_2})$ and $\sigma^2$ by projected IRLS using the score equations in Proposition~\ref{prop:score-bsigma}; accept the step by backtracking if needed to decrease $\Phi(\bb,\sigma^2;\Q)$. Keep $b_k = 0$ for $k > d_2$ to strictly enforce the rank constraint $\rank(\B) \le d_2$.
    \State \textbf{$\Q$-step:} with $(\bb,\sigma^2)$ fixed, take a Riemannian gradient step for $f(\Q)$ from Proposition~\ref{prop:stiefel}, followed by polar retraction and Armijo backtracking.
    \State \textbf{Profile $\A$:} set $\B\leftarrow \Q\diag(\bb)\Q^\top$ and update
    \[
    \A\leftarrow
    \frac{1}{J}\bm H^{1/2}
    \Pi_+^{(d_1)}
    \!\left(
    \bm H^{-1/2}\bm S_1\bm H^{-1/2}-\I_p
    \right)
    \bm H^{1/2},
    \qquad
    \bm H=\lambda_1\B+\sigma^2 \I_p.
    \]
\EndFor
\State \textbf{Return:} $(\A,\B,\sigma^2)$.
\end{algorithmic}
\end{algorithm}

\subsection{Compound symmetry}
\label{supp:cs}

\begin{proof}[Proof of Corollary~1]
For the compound symmetry temporal covariance matrix
\[
\bm{\Sigma}(\tau^2)=(1-\tau^2)I_J+\tau^2\one_J\one_J^\top,
\qquad 0\le \tau^2<1,
\]
the all-ones vector $\one_J$ is an eigenvector. Evaluating $\bm{\Sigma}(\tau^2)\one_J$ yields
\[
\bm{\Sigma}(\tau^2)\one_J = (1-\tau^2)\one_J + \tau^2(\one_J^\top \one_J)\one_J = \big(1+(J-1)\tau^2\big)\one_J.
\]
Thus, the principal eigenvalue is $\lambda_1(\tau^2)=1+(J-1)\tau^2$. Any vector $\bm{v}$ orthogonal to $\one_J$ (the contrast directions) sits in the null space of $\one_J\one_J^\top$. Thus,
\[
\bm{\Sigma}(\tau^2)\bm{v} = (1-\tau^2)\bm{v} + \tau^2\one_J(\one_J^\top \bm{v}) = (1-\tau^2)\bm{v}.
\]
Hence, every contrast direction has the exact same common eigenvalue
\[
\lambda_t(\tau^2)\equiv \lambda_c(\tau^2)=1-\tau^2, \qquad t=2,\ldots,J.
\]

Because the parameters $\lambda_c, \B$ and $\sigma^2$ are identical across all $t \ge 2$, we can factor them out of the sum. Letting $\overline{\bm S}_c = \frac{1}{J-1}\sum_{t=2}^J \bm S_t$ be the pooled contrast sample covariance, the contrast likelihood $\mathcal L_B$ becomes:
\begin{align*}
\mathcal L_B
&=
\frac{n}{2}\sum_{t=2}^J
\left\{
\log\det(\lambda_c\B+\sigma^2 I_p)
+
\tr\!\big[\bm S_t(\lambda_c\B+\sigma^2 \I_p)^{-1}\big]
\right\}\\
&=
\frac{n(J-1)}{2}
\left\{
\log\det(\lambda_c\B+\sigma^2 I_p)
+
\tr\!\big[\overline{\bm S}_c(\lambda_c\B+\sigma^2 \I_p)^{-1}\big]
\right\}.
\end{align*}

By defining a scaled loading matrix $\bm W_c = \sqrt{\lambda_c} \bm W_2$, the model contrast covariance becomes exactly $\bm W_c \bm W_c^\top + \sigma^2 I_p$. Therefore, minimizing $\mathcal L_B$ over $\bm W_c$ and $\sigma^2$ is mathematically identical to fitting a standard PPCA model \citep{tipping1999probabilistic} to the target sample covariance $\overline{\bm S}_c$ with latent dimensionality $d_2$.

Let the eigendecomposition of the pooled contrast sample covariance be
\[
\overline{\bm S}_c=\Q\diag(\bar s_1,\ldots,\bar s_p)\Q^\top,
\qquad
\bar s_1\ge\cdots\ge\bar s_p.
\]
By the global maximum likelihood estimators derived by \citet{tipping1999probabilistic}, assuming $d_2 < p$, the optimal noise variance $\widehat{\sigma}^2$ is the arithmetic average of the discarded data eigenvalues:
\[
\widehat{\sigma}^2=\frac{1}{p-d_2}\sum_{k>d_2}\bar s_k.
\]
Furthermore, the maximum likelihood estimator for the loading matrix is
\[
\widehat{\bm W}_c = \Q_{d_2} (\bm \Lambda_{d_2} - \widehat{\sigma}^2 \I_{d_2})_+^{1/2} \bm R,
\]
where $\Q_{d_2}$ contains the leading $d_2$ eigenvectors of $\overline{\bm S}_c$, $\bm \Lambda_{d_2} = \diag(\bar s_1, \ldots, \bar s_{d_2})$, $\bm R$ is an arbitrary $d_2 \times d_2$ orthogonal rotation matrix.

Returning to our original parameterization, the estimated signal covariance is $\lambda_c\widehat{\B} = \widehat{\bm W}_c\widehat{\bm W}_c^\top$. The arbitrary rotation $\bm R$ mathematically cancels out ($\bm R\bm R^\top = \I_{d_2}$), yielding a unique global solution:
\[
\lambda_c(\tau^2)\widehat{\B} = \Q_{d_2} (\bm \Lambda_{d_2} - \widehat{\sigma}^2 \I_{d_2})_+ \Q_{d_2}^\top.
\]
Dividing both sides by the scalar $\lambda_c(\tau^2) = 1-\tau^2 > 0$ isolates the optimal estimate for $\B$. Expressed via its diagonal eigenvalues $b_k$ under the basis $\Q$, we explicitly obtain $\widehat{\B}=\Q\diag(\widehat b_1,\ldots,\widehat b_p)\Q^\top$ with:
\[
\widehat b_k=
\begin{cases}
\max\!\left\{
\frac{\bar s_k-\widehat{\sigma}^2}{\lambda_c(\tau^2)},\,0
\right\}, & k=1,\ldots,d_2,\\[2mm]
0, & k=d_2+1,\ldots,p.
\end{cases}
\]
The closed-form expression for $\A^\star$ then follows directly by substituting these PPCA estimates into Theorem~1.
\end{proof}

\paragraph{Constant Profiled Contrast Likelihood}
Under Corollary~1, the estimated total contrast covariance is:
\[
\lambda_c(\tau^2)\widehat{\B}(\tau^2)+\widehat{\sigma}^2 \I_p = \Q\diag\!\big(\max(\bar s_1, \widehat{\sigma}^2), \ldots, \max(\bar s_{d_2}, \widehat{\sigma}^2), \widehat{\sigma}^2, \ldots, \widehat{\sigma}^2\big)\Q^\top.
\]
Because this optimal matrix relies solely on the data eigenvalues $\bar s_k$ and the rank $d_2$, the parameters $\lambda_c$ and $\tau^2$ completely cancel out. Consequently, the profiled contrast likelihood $\mathcal L_B$ is perfectly constant with respect to $\tau^2$ after substituting the PPCA solution. Thus, the one-dimensional profile search for $\tau^2$ can safely ignore $\mathcal L_B$ and merely minimize the mean-direction objective $\mathcal L_A\!\big(\A^\star(\tau^2)\mid \widehat{\B}(\tau^2),\widehat{\sigma}^2\big)$.

\subsection{Web Algorithm 2: Initializer 2 under Compound Symmetry}

\begin{algorithm}[H]
\caption{Initializer 2: compound-symmetry surrogate initializer}
\label{alg:init2}
\begin{algorithmic}[1]
\State \textbf{Input:} rotated covariances $\bm S_1,\overline{\bm S}_c$, ranks $(d_1,d_2)$, number of visits $J$, and boundary tolerance $\epsilon_\tau>0$.
\State Compute the eigendecomposition
\[
\overline{\bm S}_c=\Q\diag(\bar s_1,\ldots,\bar s_p)\Q^\top,
\qquad
\bar s_1\ge\cdots\ge\bar s_p.
\]
\State Set
\[
\widehat{\sigma}^2=\frac{1}{p-d_2}\sum_{k>d_2}\bar s_k.
\]
\For{each candidate $\tau^2\in[0,1-\epsilon_\tau]$ in a bounded one-dimensional search}
    \State Compute
    \[
    \lambda_1(\tau^2)=1+(J-1)\tau^2,
    \qquad
    \lambda_c(\tau^2)=1-\tau^2.
    \]
    \State Set
    \[
    \widehat b_k(\tau^2)=
    \begin{cases}
    \max\!\left\{
    \frac{\bar s_k-\widehat{\sigma}^2}{\lambda_c(\tau^2)},\,0
    \right\}, & k=1,\ldots,d_2,\\[2mm]
    0, & k=d_2+1,\ldots,p.
    \end{cases}
    \]
    \State Form matrix $\widehat{\B}(\tau^2)=\Q\diag(\widehat b_1(\tau^2),\ldots,\widehat b_p(\tau^2))\Q^\top.$
    \State Update
    \[
    \A^\star(\tau^2)
    =
    \frac{1}{J}\bm H(\tau^2)^{1/2}
    \Pi_+^{(d_1)}
    \!\left(
    \bm H(\tau^2)^{-1/2}\bm S_1\bm H(\tau^2)^{-1/2}-I_p
    \right)
    \bm H(\tau^2)^{1/2},
    \]
    where
    \[
    \bm H(\tau^2)=\lambda_1(\tau^2)\widehat{\B}(\tau^2)+\widehat{\sigma}^2 I_p.
    \]
    \State Evaluate the profiled mean-direction objective
    \[
    \mathcal L_A\!\big(\A^\star(\tau^2)\mid \widehat{\B}(\tau^2),\widehat{\sigma}^2\big).
    \]
\EndFor
\State Return the value of $\tau^2$ minimizing the profiled objective and the corresponding
$\big(\A^\star,\widehat{\B},\widehat{\sigma}^2\big)$.
\end{algorithmic}
\end{algorithm}

\webappendixsection{Additional Simulation Details}

\label{webapp:sim_details}

This Web Appendix provides extended results for the simulation studies presented in Section~3 of the main text. Specifically, we present the robustness of the proposed HPPCA framework under a discrete-time LDS with a longer panel in \Cref{webapp:sec:ar1_j10}. We extend the misspecification analysis to a second-order autoregressive (AR(2)) generating process in \Cref{webapp:sec:ar2}. Finally, \Cref{webapp:sec:efficiency} provides comprehensive numerical results evaluating the choice of initializers.

\subsection{Robustness to Misspecification: AR(1) Dynamics with Longer Panels}
\label{webapp:sec:ar1_j10}

In Section~3.3, we evaluated the robustness of HPPCA when the true underlying dynamics follow a discrete-time AR(1) process, rather than the assumed continuous-time Gaussian process. The main text demonstrated these results for a panel with $J=5$ visits and a latent dimensionality of $d=4$. Here, we present the corresponding results for a longer panel ($J=10$) and a higher latent dimensionality ($d=8$, partitioned into $d_1=4$ subject-level static traits and $d_2=4$ dynamic traits). The generative mechanism remains identical to that described in Section~3.3. \Cref{fig:sim_ar1_j10} displays the imputation MSE for the $J=10$ scenario. The empirical trends mirror those observed in the $J=5$ setting, but the performance gaps between HPPCA and the baseline methods become more pronounced. 

\begin{figure}[!htbp]
    \centering
    \includegraphics[width=0.9\linewidth]{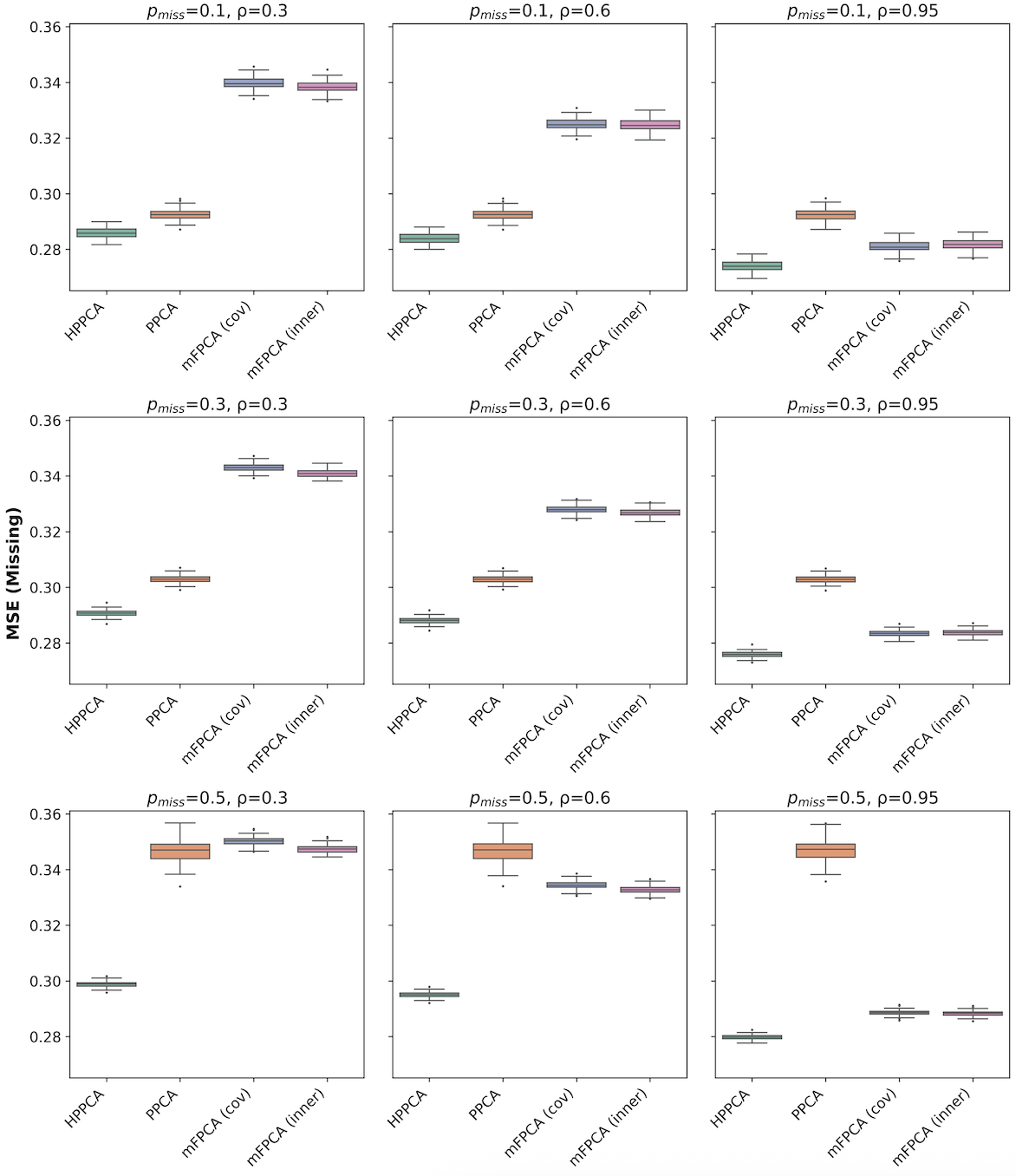}
    \caption{\textbf{Predictive performance under the discrete LDS generative mechanism ($J=10, d=8$).} Boxplots display the imputation MSE for HPPCA and the baseline methods (PPCA, mFPCA(cov), mFPCA(inner)). The grid varies the missingness probability $p_{\text{miss}} \in \{0.1, 0.3, 0.5\}$ (rows) and the temporal autocorrelation parameter $\rho \in \{0.3, 0.6, 0.95\}$ (columns).}
    \label{fig:sim_ar1_j10}
\end{figure}

\subsection{Robustness to Misspecification: AR(2) Generative Process}
\label{webapp:sec:ar2}

To further evaluate the robustness of HPPCA against higher-order discrete state transitions, we modified the LDS generative mechanism to follow a second-order autoregressive (AR(2)) process. An AR(2) process can exhibit more complex pseudo-cyclical or delayed momentum dynamics compared to an AR(1) process, deviating further from the continuous-time GP assumptions. Meanwhile, we did not partition the latent space here as in \Cref{webapp:sec:ar1_j10} to further illustrate the method's generalizability.

For subject $i$ at time step $j$, the latent factors were generated as:
\[
\bm{Z}_{ij} = \bm{\Phi}_1 \bm{Z}_{i(j-1)} + \bm{\Phi}_2 \bm{Z}_{i(j-2)} + \bm{w}_{ij},
\]
where the transition matrices were defined as diagonal matrices $\bm{\Phi}_1 = \phi_{1} \bm{I}_{d}$ and $\bm{\Phi}_2 = \phi_{2} \bm{I}_{d}$, where $\phi_1 = 0.8, \phi_2 = 0.1$ \citep{shen2024principal}. The process noise $\bm{w}_{ij} \sim \mathcal{N}(\bm{0}, \sigma_w^2\bm{I}_{d})$, where $\sigma_w^2$ is the variance of the latent noise, was analytically scaled by solving the discrete Yule-Walker equations to maintain marginal unit variance $\bm{Z}_{ij} \sim \mathcal{N}(\bm{0}, \bm{I}_d)$ across all time steps, preserving parameter identifiability. Finally, these latent states were mapped to the high-dimensional observation space following the model in Section~3.3: $\bm{Y}_{ij} = \bm{W}\bm{Z}_{ij} + \sigma \bm{\varepsilon}_{ij}, \bm{\varepsilon}_{ij} \sim N(\bm{0}, \bm{I}_p)$, where the true entries in loadings $\bm{W} \in \mathbb{R}^{p \times d}$ were generated from a standard normal distribution, and $\sigma^2$ was fixed at 0.25. Similarly, all observation times were uniformly scaled to the interval $[0, 1]$ and evenly distributed. 

\Cref{fig:sim_ar2} illustrates the imputation MSE under the AR(2) generative model. Consistent with the AR(1) results, HPPCA consistently yields the lowest imputation MSE compared to PPCA, mFPCA(cov) and mFPCA(inner), especially when the missing rate is high (50\%) and the time panel is long ($J = 10$). This demonstrates that the framework's reliance on continuous-time GP priors acts as a highly effective and robust regularizer. HPPCA successfully accommodates the higher-order discrete dynamics.

\begin{figure}[!htbp]
    \centering
    \includegraphics[width=0.6\linewidth]{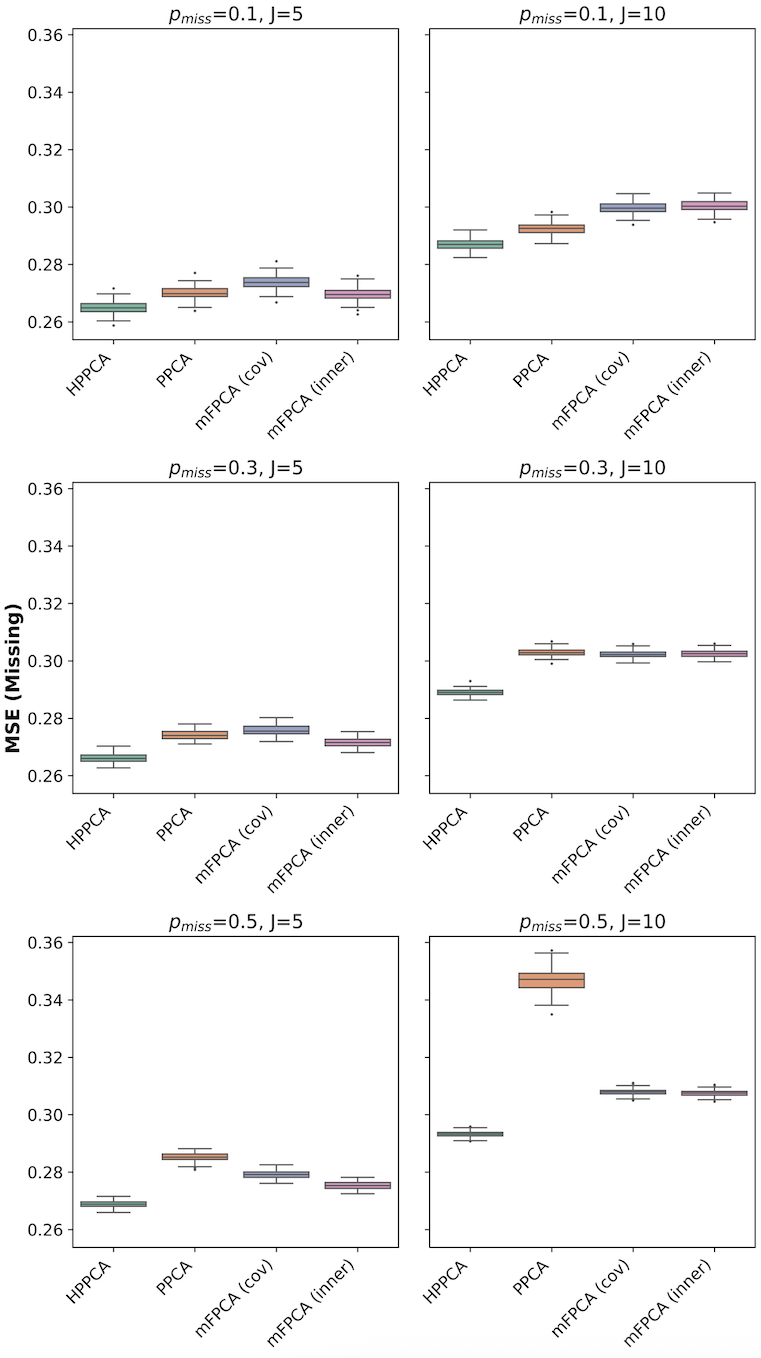}
    \caption{\textbf{Predictive performance under an AR(2) generative mechanism.} Boxplots display the imputation MSE for HPPCA and the baseline methods (PPCA, mFPCA(cov), mFPCA(inner)). The grid varies the missingness probability $p_{\text{miss}} \in \{0.1, 0.3, 0.5\}$ (rows) and the number of visits $J \in \{5, 10 \}$ (columns).}
    \label{fig:sim_ar2}
\end{figure}

\subsection{Computational Efficiency of the Initializers}
\label{webapp:sec:efficiency}

Observational data frequently suffer from item-level missingness, varying numbers of visits, and smooth, non-uniform temporal dynamics (e.g., an RBF covariance structure). To apply the initializers in practice, we implemented a computationally efficient empirical pre-processing procedure. 

To start the initializer, any missing data is first filled using participant-specific feature mean imputation; if that is not available, the gap is filled using the global feature mean. For Initializer 1, we estimate a rough temporal covariance from pairwise observed-data differences. To estimate a known temporal covariance strictly satisfying Assumption 1, we project this empirical matrix onto the valid parameter space such that its principal eigenvector is exactly parallel to the mean direction $\one_J$, while maintaining symmetry, positive semi-definiteness, and a unit diagonal. Initializer 2 completely bypasses this projection step, since we just assume the temporal covariance matrix follow the compound symmetry structure.

\Cref{tab:init_5_2_2,tab:init_10_4_4} demonstrate the substantial computational benefits provided by these initializers over random initialization. These results confirm that both Initializer 1 and 2 robustly overcome complex data challenges to provide computationally cheap starting points for the HPPCA framework.

\begin{table}[!htbp]
  \centering
  \caption{\textbf{Convergence Efficiency and Parameter Proximity with Different Intialization Methods ($J=5, d_1=2, d_2=2$).} Comparison of the proposed initializers against random initialization across varying missing data rates. Results are based on 100 independent replicates where data were generated using an RBF kernel ($\ell=10$) under the HPPCA model as in Section 3.1. We report the mean (standard deviation) for the principal angles between the estimated and true subspaces for $\bm W_1$ and $\bm W_2$ at both initialization and convergence. Final estimated bias (standard deviation) for $\sigma^2$ and $\ell$ are provided alongside computational cost, measured by total EM iterations and wall-clock time (minutes).}
  \label{tab:init_5_2_2}
  \vspace{0.2cm}
  \resizebox{\textwidth}{!}{
    \begin{tabular}{llcccccccc}
    \toprule
    \textbf{Missing} & \textbf{Init} & \textbf{Init $\bm W_1$} & \textbf{Final $\bm W_1$} & \textbf{Init $\bm W_2$} & \textbf{Final $\bm W_2$} & \textbf{Final $\sigma^2$} & \textbf{Final $\ell$} & \textbf{Final} & \textbf{Time} \\
    \textbf{Rate} & \textbf{Method} & \textbf{Angle ($^\circ$)} & \textbf{Angle ($^\circ$)} & \textbf{Angle ($^\circ$)} & \textbf{Angle ($^\circ$)} & \textbf{Estimate} & \textbf{Estimate} & \textbf{Iters} & \textbf{(min)} \\
    \midrule
    \multirow{3}{*}{$0\%$}  
      & Random & 85.74 (3.09) & 2.03 (0.80) & 85.66 (3.54) & 0.47 (0.06) & 0.00 (0.00) & 0.01 (0.08) & 5459 (1011) & 32.6 (12.9) \\
    & Initializer 1 & 2.10 (0.78) & 2.03 (0.80) & 0.61 (0.07) & 0.47 (0.06) & 0.00 (0.00) & 0.01 (0.08) & 2760 (869)  & 16.8 (8.5)  \\
    & Initializer 2 & 2.10 (0.78) & 2.03 (0.80) & 2.63 (0.22) & 0.47 (0.06) & 0.00 (0.00) & 0.01 (0.08) & 2748 (870)  & 16.1 (8.2)  \\
    \cmidrule{1-10}
    \multirow{3}{*}{$10\%$}  
      & Random & 85.74 (3.09) & 2.03 (0.80) & 85.66 (3.54) & 0.50 (0.07) & 0.00 (0.00) & 0.01 (0.08) & 5496 (1245) & 38.2 (16.8) \\
    & Initializer 1 & 2.15 (0.77) & 2.03 (0.80) & 1.11 (0.14) & 0.50 (0.07) & 0.00 (0.00) & 0.01 (0.08) & 2839 (985)  & 20.6 (10.5) \\
    & Initializer 2 & 2.15 (0.77) & 2.03 (0.80) & 4.12 (0.24) & 0.50 (0.07) & 0.00 (0.00) & 0.01 (0.08) & 2844 (985)  & 21.1 (11.2) \\
    \cmidrule{1-10}
    \multirow{3}{*}{$30\%$} 
      & Random & 85.74 (3.09) & 2.06 (0.80) & 85.66 (3.54) & 0.57 (0.08) & 0.00 (0.00) & 0.01 (0.08) & 5348 (803)  & 40.2 (16.9) \\
    & Initializer 1 & 2.34 (0.76) & 2.06 (0.80) & 2.74 (0.58) & 0.57 (0.08) & 0.00 (0.00) & 0.01 (0.08) & 2748 (654)  & 21.0 (10.0) \\
    & Initializer 2 & 2.34 (0.77) & 2.06 (0.80) & 8.77 (0.33) & 0.57 (0.08) & 0.00 (0.00) & 0.01 (0.08) & 2756 (663)  & 20.5 (9.0)  \\
    \cmidrule{1-10}
    \multirow{3}{*}{$50\%$} 
      & Random & 85.74 (3.09) & 2.09 (0.78) & 85.66 (3.54) & 0.68 (0.09) & -0.00 (0.00) & 0.01 (0.08) & 5550 (1044) & 44.0 (20.7) \\
    & Initializer 1 & 2.83 (0.80) & 2.09 (0.78) & 6.48 (1.90) & 0.68 (0.09) & -0.00 (0.00) & 0.01 (0.08) & 3772 (937)  & 29.7 (14.9) \\
    & Initializer 2 & 2.83 (0.80) & 2.09 (0.78) & 18.30 (0.61) & 0.68 (0.09) & -0.00 (0.00) & 0.01 (0.08) & 3798 (927)  & 29.4 (14.3) \\
    \bottomrule
    \end{tabular}%
  }
\end{table}

\begin{table}[!htbp]
  \centering
  \caption{\textbf{Convergence Efficiency and Parameter Proximity with Different Intialization Methods ($J=10, d_1=4, d_2=4$).} Comparison of the proposed initializers against random initialization across varying missing data rates. Results are based on 100 independent replicates where data were generated using an RBF kernel ($\ell=10$) under the HPPCA model as in Section 3.1. We report the mean (standard deviation) for the principal angles between the estimated and true subspaces for $\bm W_1$ and $\bm W_2$ at both initialization and convergence. Final estimated bias (standard deviation) for $\sigma^2$ and $\ell$ are provided alongside computational cost, measured by total EM iterations and wall-clock time (minutes).}
  \label{tab:init_10_4_4}
  \vspace{0.2cm}
  \resizebox{\textwidth}{!}{
    \begin{tabular}{llcccccccc}
    \toprule
    \textbf{Missing} & \textbf{Init} & \textbf{Init $\bm W_1$} & \textbf{Final $\bm W_1$} & \textbf{Init $\bm W_2$} & \textbf{Final $\bm W_2$} & \textbf{Final $\sigma^2$} & \textbf{Final $\ell$} & \textbf{Final} & \textbf{Time} \\
    \textbf{Rate} & \textbf{Method} & \textbf{Angle ($^\circ$)} & \textbf{Angle ($^\circ$)} & \textbf{Angle ($^\circ$)} & \textbf{Angle ($^\circ$)} & \textbf{Estimate} & \textbf{Estimate} & \textbf{Iters} & \textbf{(min)} \\
    \midrule
    \multirow{3}{*}{$0\%$}  
      & Random & 87.21 (2.20) & 2.71 (0.57) & 86.87 (2.37) & 0.38 (0.04) & 0.00 (0.00) & 0.00 (0.03) & 11048 (2449) & 141.4 (51.9) \\
    & Initializer 1 & 18.94 (33.49) & 2.71 (0.57) & 16.96 (34.30) & 0.38 (0.04) & 0.00 (0.00) & 0.00 (0.03) & 5520 (3218)  & 68.0 (45.0)  \\
    & Initializer 2 & 2.78 (0.64) & 2.71 (0.57) & 0.43 (0.05) & 0.38 (0.04) & 0.00 (0.00) & 0.00 (0.03) & 4236 (1622)  & 52.4 (26.5)  \\
    \cmidrule{1-10}
    \multirow{3}{*}{$10\%$}  
      & Random & 87.21 (2.20) & 2.72 (0.57) & 86.87 (2.37) & 0.40 (0.05) & 0.00 (0.00) & 0.00 (0.03) & 10949 (2466) & 175.0 (64.6) \\
    & Initializer 1 & 12.87 (27.27) & 2.72 (0.57) & 11.30 (28.00) & 0.40 (0.05) & 0.00 (0.00) & 0.00 (0.03) & 5040 (2826)  & 77.0 (54.8)  \\
    & Initializer 2 & 2.82 (0.65) & 2.72 (0.57) & 1.01 (0.16) & 0.40 (0.05) & 0.00 (0.00) & 0.00 (0.03) & 4257 (1607)  & 69.3 (33.6)  \\
    \cmidrule{1-10}
    \multirow{3}{*}{$30\%$} 
      & Random & 87.21 (2.20) & 2.72 (0.57) & 86.87 (2.37) & 0.46 (0.05) & 0.00 (0.00) & 0.00 (0.03) & 11489 (2798) & 185.3 (81.5) \\
    & Initializer 1 & 4.60 (11.75) & 2.72 (0.57) & 4.60 (11.82) & 0.46 (0.05) & 0.00 (0.00) & 0.00 (0.03) & 5130 (2467)  & 83.3 (51.5)  \\
    & Initializer 2 & 2.94 (0.62) & 2.72 (0.57) & 3.19 (0.77) & 0.46 (0.05) & 0.00 (0.00) & 0.00 (0.03) & 5009 (2078)  & 81.0 (46.1)  \\
    \cmidrule{1-10}
    \multirow{3}{*}{$50\%$} 
      & Random & 87.21 (2.20) & 2.74 (0.57) & 86.87 (2.37) & 0.56 (0.06) & 0.00 (0.00) & 0.00 (0.03) & 11263 (2072) & 189.2 (68.7) \\
    & Initializer 1 & 9.87 (22.67) & 2.74 (0.57) & 13.03 (21.81) & 0.56 (0.06) & 0.00 (0.00) & 0.00 (0.03) & 5604 (2193)  & 87.9 (38.1)  \\
    & Initializer 2 & 3.26 (0.66) & 2.74 (0.57) & 17.83 (0.32) & 0.56 (0.06) & 0.00 (0.00) & 0.00 (0.03) & 5075 (1384)  & 83.7 (36.6)  \\
    \bottomrule
    \end{tabular}%
  }
\end{table}

\webappendixsection{Additional Details for the RECOVER Application}

This appendix provides supplementary information regarding the data preprocessing, exploratory data distributions, physical symptom definitions, model parameter estimation diagnostics, subspace orthogonality, and secondary downstream prediction metrics for the RECOVER adult cohort analysis presented in Section~4 of the main text.

\subsection{Cohort Characteristics}

A major challenge in analyzing the RECOVER cohort is the extreme irregularity and sparsity of the follow-up schedules. \Cref{fig:recover_distributions} illustrates the empirical distributions of the longitudinal data. The number of visits per participant varies widely from $1$ to $16$ (\Cref{fig:recover_distributions:visits}), and the timing of these visits spans from the acute infection phase to over $1{,}400$ days post-infection (\Cref{fig:recover_distributions:days}). This highlights the necessity of the continuous-time, subject-specific generative mechanism proposed in the HPPCA framework.

\begin{figure}[!htbp]
    \centering
    \begin{subfigure}{0.43\textwidth}
        \centering
        \includegraphics[width=\linewidth]{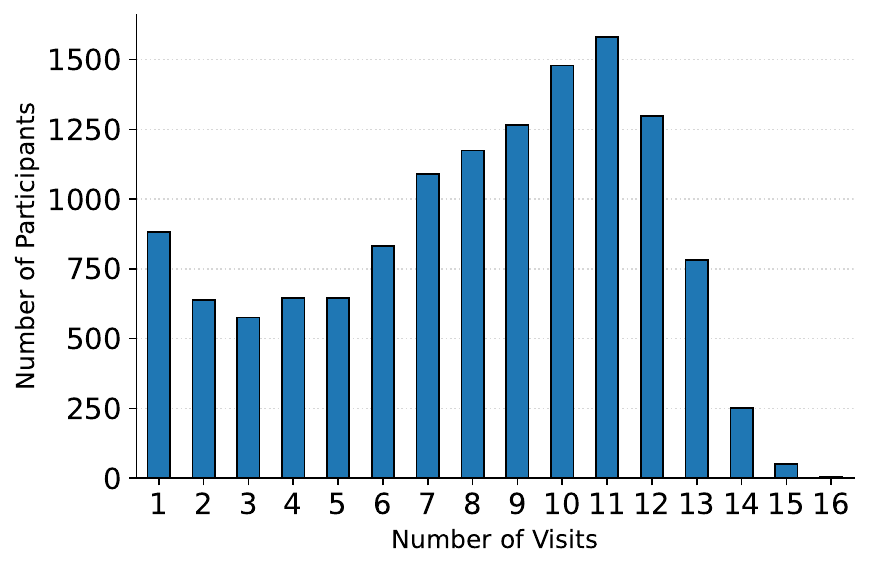}
        \caption{Distribution of Total Visits per Participant}
        \label{fig:recover_distributions:visits}
    \end{subfigure} \hfill
    \begin{subfigure}{0.47\textwidth}
        \centering
        \includegraphics[width=\linewidth]{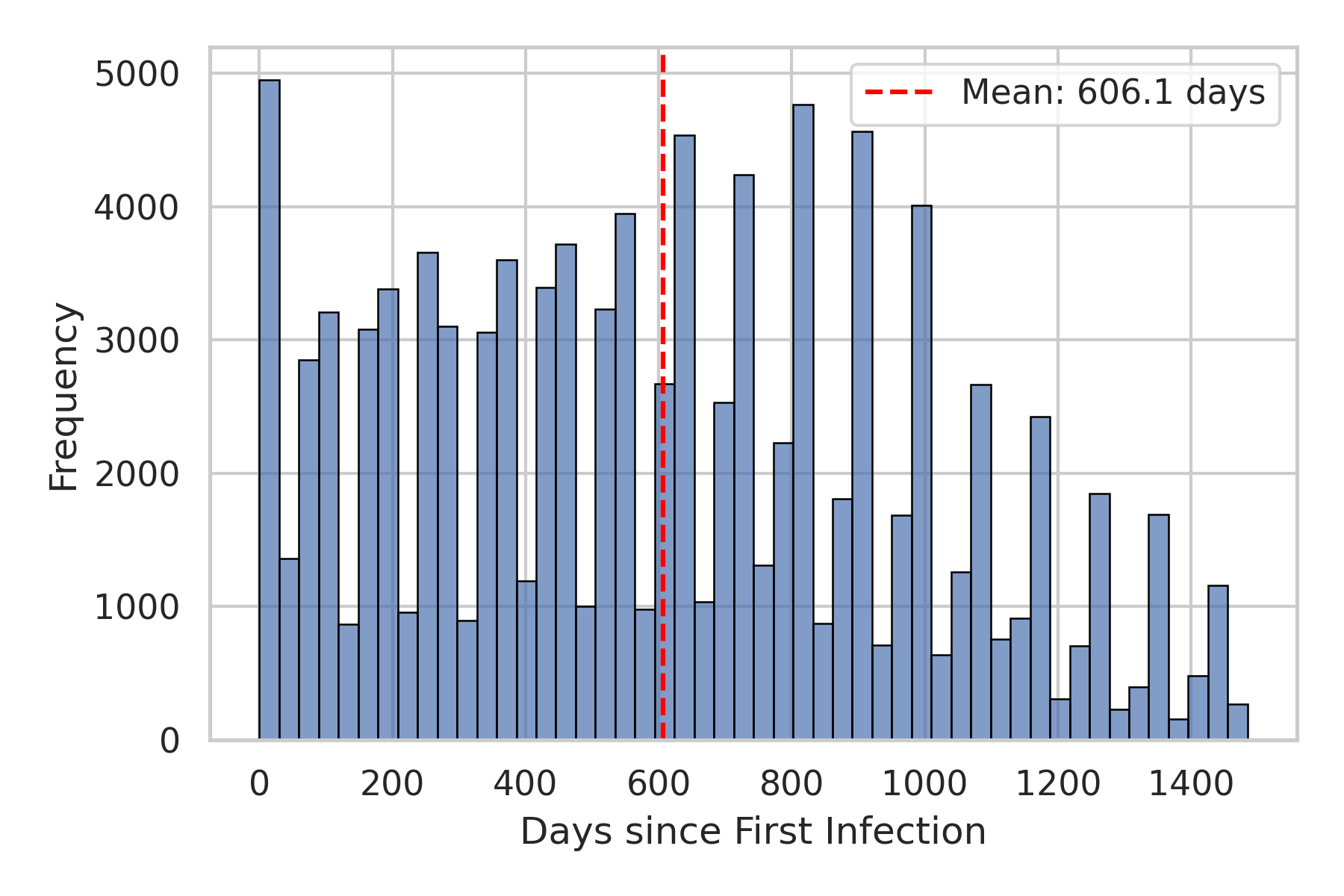}
        \caption{Distribution of Visit Days since First Infection}
        \label{fig:recover_distributions:days}
    \end{subfigure}
    \caption{\textbf{Distributions of Longitudinal Follow-ups in the RECOVER Cohort.} (a) Distribution of the total number of visits ($J_i$) per participant, highlighting the highly variable sequence lengths. (b) Distribution of observation times ($t_{ij}$), measured in days since the participant's first documented SARS-CoV-2 infection (Mean = $606.1$ days).}
    \label{fig:recover_distributions}
\end{figure}

\subsection{Physical Symptoms}
The original physical symptom questionnaire included 59 burden variables. To ensure a uniform feature space across all participants, we excluded two sex-specific burden variables (on menopause symptoms and menstrual cycle respectively), resulting in $p=57$ physical symptom features. \Cref{tab:symptom_dictionary} provides the mapping of the $p=57$ physical symptom burden variable labels utilized in the analysis, including their physiological categorization and the survey questions administered to participants. As mentioned in Section~4.1, we employed a deterministic imputation strategy for these variables based on survey logic. Specifically, missing burden scores were imputed as ``Not at all'' (score 1) if a participant answered ``No, I have NOT had this symptom'' to the corresponding gating question (``Please tell us at what time(s) you have had any of the following symptoms''). While this imputation generally followed a 1-to-1 mapping between a symptom and its burden, two categories rely on broader screening questions: neurologic variables were mapped to the umbrella symptom ``Nerve problems (tremor, shaking, abnormal movements, numbness, tingling, burning, can't move part of body, new seizures)'', and all pain variables (except headaches) were mapped to ``Pain in any part of your body''.

\footnotesize{
\begin{longtable}{p{0.15\textwidth} p{0.2\textwidth} p{0.57\textwidth}}
\caption{\textbf{Dictionary of Physical Symptom Variables in the RECOVER Cohort}} \label{tab:symptom_dictionary} \\
\toprule

\textbf{Category} & \textbf{Symptom Label} & \textbf{Survey Question} \\
\midrule
\endfirsthead

\multicolumn{3}{c}%
{{\bfseries Web \tablename\ \thetable{} -- continued from previous page}} \\
\toprule
\textbf{Category} & \textbf{Symptom Label} & \textbf{Survey Question} \\
\midrule
\endhead

\midrule \multicolumn{3}{r}{{Continued on next page}} \\
\endfoot

\bottomrule
\endlastfoot

Gastrointestinal & Being Full (Bloated) & How much does this feeling of being full (bloated) bother you? \\
Gastrointestinal & Constipation & How much does the constipation bother you? \\
Gastrointestinal & Abdominal Pain & How much does this cramping or colicky abdominal pain bother you? \\
Gastrointestinal & Diarrhea & How much does the diarrhea bother you? \\
Gastrointestinal & Belly Symptoms & How much do your belly symptoms bother you? \\
Gastrointestinal & Nausea & How much does the nausea bother you? \\
Gastrointestinal & Reflux / Heartburn & How much does the reflux or heartburn bother you? \\
Gastrointestinal & Vomiting & How much does this vomiting bother you? \\
Pain & Pelvic Pain & How much does your pelvic or genital pain bother you? \\
Pain & Belly Pain & How much does your abdominal (belly) pain bother you? \\
Pain & Back Pain & How much does your back or spinal pain bother you? \\
Pain & Chest Pain & How much does your chest pain bother you? \\
Pain & Foot Pain & How much does your foot pain bother you? \\
Pain & Headaches & How much do your headaches bother you? \\
Pain & Joint Pain & How much does your joint pain bother you? \\
Pain & Mouth Pain & How much does your mouth pain bother you? \\
Pain & Muscle Pain & How much does your muscle pain bother you? \\
Pain & Soreness After Activities & How much does your soreness or fatigue after non-strenuous, everyday activities bother you? \\
Pain & Skin Pain & How much does your skin pain bother you? \\
Pain & Throat Pain & How much does your throat pain bother you? \\
Neurologic & Faint / Dizzy / Goofy & How much does feeling faint, dizzy, or goofy bother you? \\
Neurologic & Abnormal Movements & How much do your abnormal movements bother you? \\
Neurologic & Brain Fog & How much do your problems thinking or concentrating ("brain fog") bother you? \\
Neurologic & Numbness / Tingling / Burning & How much does your nerve numbness, tingling, or burning bother you? \\
Neurologic & Seizures & How much do your seizures bother you? \\
Neurologic & Sleep Problems & How much do your sleep problems bother you? \\
Neurologic & Tremors & How much do your tremors bother you? \\
Respiratory & Chronic Cough & How much does your persistent (chronic) cough bother you? \\
Respiratory & Shortness of Breath & How much does your shortness of breath bother you? \\
Respiratory & Wheezing & How much does your wheezing or whistling in your chest bother you? \\
Cardiovascular & Palpitations / Racing Heart & How much do your palpitations, racing heart, arrhythmia, or skipped beats bother you? \\
Cardiovascular & Swelling of Legs & How much does the swelling of your legs bother you? \\
Systemic & Excessive Thirst & How much does your excessive thirst bother you? \\
Systemic & Fever / Chills / Sweats & How much do your fever, chills, sweats (flu-like symptoms) bother you? \\
Systemic & Flushing & How much does your flushing bother you (a sudden feeling of warmth and reddening of the face)? \\
Systemic & Fatigue & How much does your fatigue bother you? \\
Systemic & Post-Exertional Malaise & How much does your post-exertional malaise bother you? \\
Systemic & Feeling Hot or Cold & How much does feeling hot or cold for no reason bother you? \\
Systemic & Weakness in Arms or Legs & How much does the weakness in your arms or legs bother you? \\
Systemic & Cold Limbs & How much does having cold limbs bother you? \\
Sensory & Smell / Taste Change & How much does your loss of or change in smell or taste bother you? \\
Sensory & Hearing Loss & How much do your problems with hearing (hearing loss or ringing in ears) bother you? \\
Sensory & Sinus Problems & How much does having a runny nose or sinus problems bother you? \\
Sensory & Smells Making You Sick & How much does having some smells, foods, medications, or chemicals making you feel sick bother you? \\
Sensory & Teeth Problems & How much do your problems with teeth bother you? \\
Sensory & Vision Problems & How much do your vision problems bother you? \\
Genitourinary & Bladder Problems & How much do your bladder problems bother you? \\
Genitourinary & Fertility Changes & How much do the changes in your fertility or difficulty getting pregnant bother you? \\
Genitourinary & Sex Changes & How much do your changes in desire for, comfort with, or capacity for sex bother you? \\
Dermatologic & Skin Color Change & How much does the change in your skin color bother you? \\
Dermatologic & Dry Eyes & How much do your excessively dry eyes bother you? \\
Dermatologic & Dry Mouth & How much does your excessively dry mouth bother you? \\
Dermatologic & Allergic Reactions & How much do your severe allergic reactions bother you? \\
Dermatologic & Hair Loss & How much does your hair loss bother you? \\
Dermatologic & Hives & How much do hives (skin redness or swelling) bother you? \\
Dermatologic & Itching & How much does itching bother you? \\
Dermatologic & Skin Rash & How much does your skin rash bother you? \\

\end{longtable}
}

\normalsize
\subsection{HPPCA Parameter Trends}
To evaluate whether the model captures more data variation when the latent dimension increases, we tracked the estimated parameters across varying latent capacities ($d_1 = d_2 \in \{1,\dots,10\}$). As shown in \Cref{fig:hppca_trends}, increasing the latent dimension monotonically decreases both the estimated measurement noise variance $\sigma^2$ (left panel) and the training reconstruction MSE (right panel), indicating that the latent components are successfully absorbing structured signals.

\begin{figure}[!htbp]
    \centering
    \includegraphics[width=\linewidth]{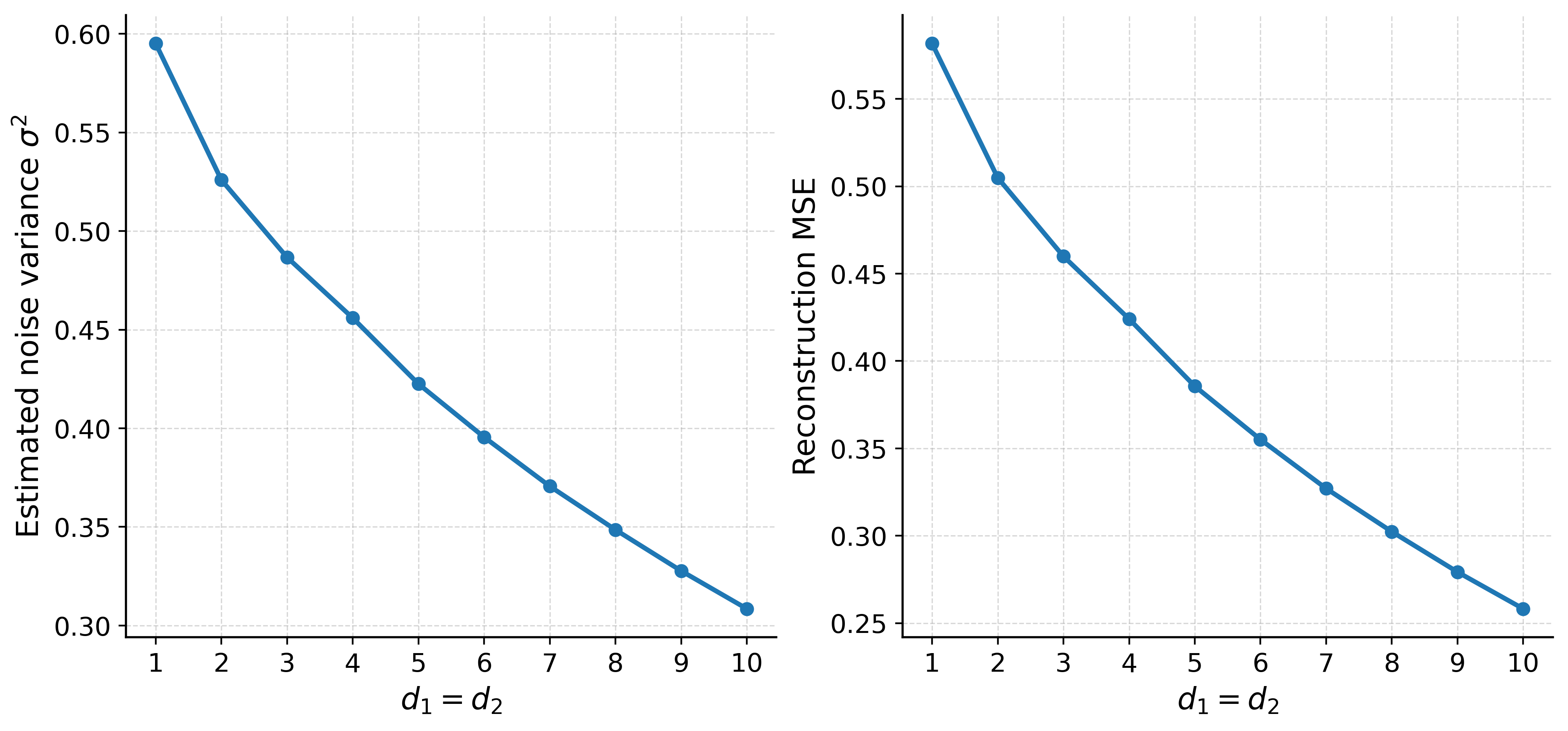}
    \caption{\textbf{HPPCA Parameter Estimates across Latent Dimensions.} Estimated isotropic noise variance $\sigma^2$ (left), and overall reconstruction MSE (right) as a function of the latent dimension capacity $d_1=d_2$.}
    \label{fig:hppca_trends}
\end{figure}

\subsection{Subspace Orthogonality Diagnostics}
To verify that HPPCA extracts diverse clinical features within each hierarchical level, we evaluated the pairwise angles between the internal column vectors of the estimated static loading matrix $\widehat{\bm{W}}_1$, and separately for the dynamic loading matrix $\widehat{\bm{W}}_2$, at $d_1 = d_2 = 10$. As shown in \Cref{fig:angle_heatmap}, these intra-matrix pairwise angles range between 55° and 90°. Because 90° indicates perfect orthogonality, this result confirms that the individual components within each loading matrix capture distinct, non-overlapping dimensions of data complexity, ensuring that the 10-dimensional latent capacities are utilized efficiently without component collapse.

\begin{figure}[htbp]
    \centering
    \begin{subfigure}{0.49\linewidth}
    \centering
\includegraphics[width=\linewidth]{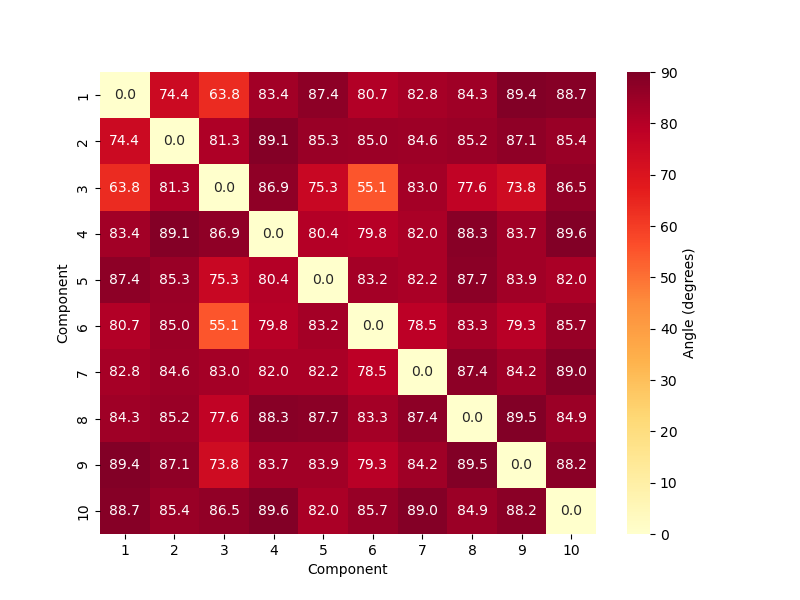}
    \caption{Angles Between $\bm{W}_1$ Loadings}
    \end{subfigure}
    \begin{subfigure}{0.49\linewidth}
    \centering
\includegraphics[width=\linewidth]{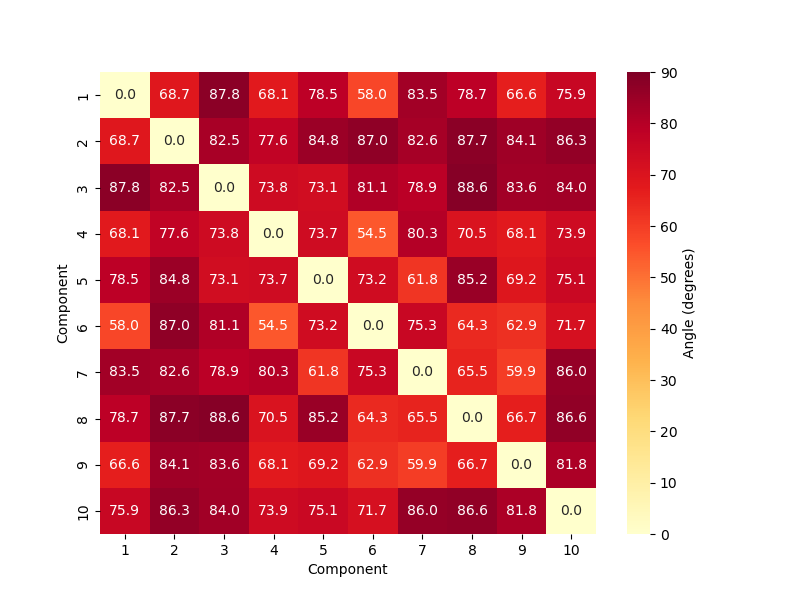}
    \caption{Angles Between $\bm{W}_2$ Loadings}
    \end{subfigure}
    \caption{\textbf{Angles Between Static and Dynamic Subspaces.} Heatmap of the pairwise angles (in degrees) between the $10$ components of the static loading matrix ($\bm{W}_1$) and the $10$ components of the dynamic loading matrix ($\bm{W}_2$) at $d_1=d_2=10$. Values near $90^\circ$ (dark red) indicate orthogonal, independent directions.}
    \label{fig:angle_heatmap}
\end{figure}

\subsection{Secondary Downstream Prediction Metrics}
\begin{figure}[htbp]
    \centering
    \includegraphics[width=0.9\linewidth]{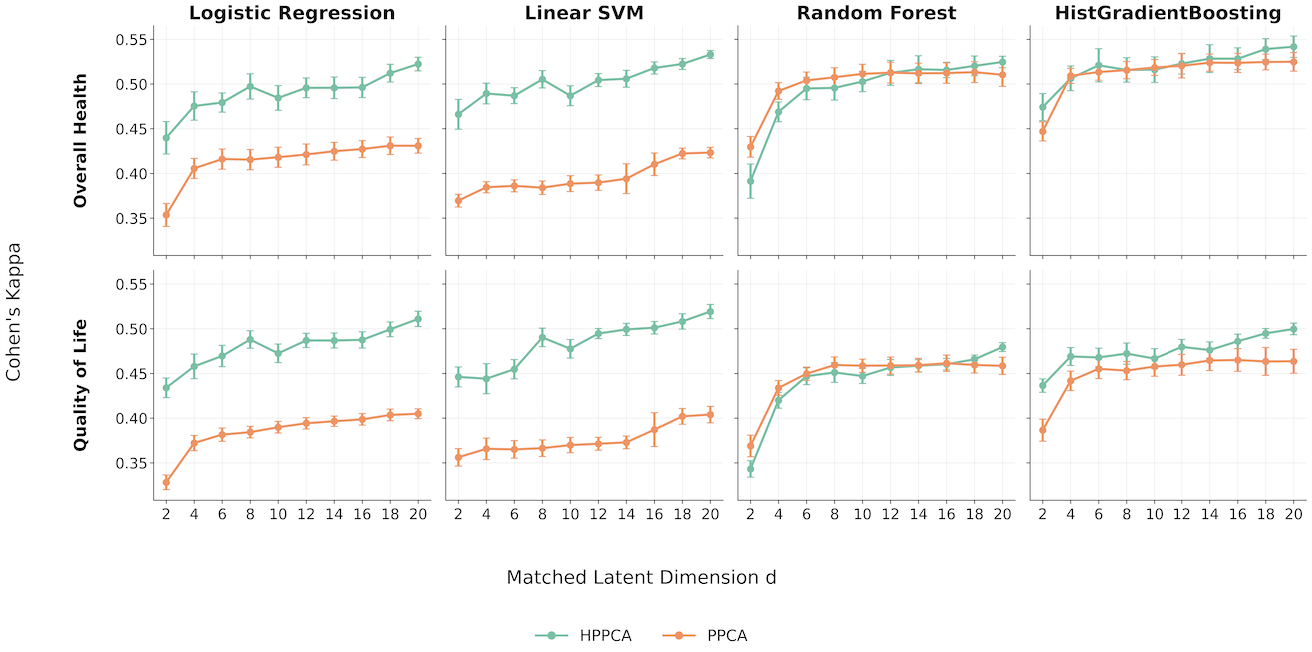}
    \caption{\textbf{Secondary Metric: Cohen's Kappa Score for Subject-Visit Level Prediction.} Cohen's kappa score for predicting the 5-category ordinal PROMIS-01 global health score at specific longitudinal visits. HPPCA joint embeddings ($\bm{Z}_{1,i} \oplus \bm{Z}_{2,ij}$) systematically outperform PPCA embeddings with a matched total rank. Error bars represent standard deviations across 5-fold cross-validation.}
    \label{fig:visit_kappa}
\end{figure}
In Section~4 of the main text, we presented balanced accuracy as the primary metric to evaluate the predictive performance of latent representations on two visit-level outcomes: overall health and quality of life. To provide a comprehensive evaluation, we also calculated Cohen's Kappa, as shown in \Cref{fig:visit_kappa}. For both outcomes, the joint embeddings of HPPCA yielded substantially higher Cohen's kappa scores across logistic regression, linear SVM, and gradient boosting models compared to PPCA. When using a random forest classifier, Cohen's kappa scores were similar between the two methods at lower matched dimensions, but increased more for HPPCA as the latent dimension ($d$) increased to 20. These findings are highly consistent with our balanced accuracy results, further validating our conclusions.

\bibliographystyle{plainnat} 
\bibliography{arXiv-version/combined_ref}
\end{document}